\documentclass[11pt]{article}
\usepackage[a4paper,twoside,margin=2cm]{geometry}
\usepackage[utf8]{inputenc}
\usepackage[T1]{fontenc}
\usepackage{amsmath,amsfonts,amssymb,amsthm}
\usepackage{mathrsfs}
\usepackage{graphicx}
\usepackage{hyperref}
\usepackage{url} 
\usepackage{color}
\usepackage{cite}
\usepackage{xfrac}

\usepackage{chngcntr}
\counterwithin{figure}{section}

\theoremstyle{plain}
\newtheorem{theorem}{Theorem}[section]
\newtheorem{proposition}[theorem]{Proposition}
\newtheorem{lemma}[theorem]{Lemma}

\newcommand{\R}{\mathbb{R}}
\newcommand{\BB}{\mathcal{B}}
\newcommand{\LL}{\mathcal{L}}
\newcommand{\complex}{\mathscr{C}}
\newcommand{\surf}{\mathscr{S}}
\newcommand{\poly}{\text{poly}}

\newcommand{\mmid}{\text{mid}}
\newcommand{\half}{\ensuremath{\sfrac12}}

\definecolor{blueblack}{rgb}{0,0,.7}
\newcommand{\emphdef}[1]{%
  \textcolor{blueblack}{%
    \textbf{\emph{#1}}%
  }}

\title{An FPT Algorithm for the Embeddability of Graphs\\ into Two-Dimensional Simplicial Complexes\thanks{Partially supported by the ANR projects Min-Max (ANR-19-CE40-0014) and SoS (ANR-17-CE40-0033).  This version is to appear in \emph{SIAM Journal on Computing}.  A preliminary version appeared in \emph{Proceedings of the European Symposium on Algorithms 2021}.}}

\author{\'Eric Colin de Verdi\`ere\thanks{LIGM, CNRS, Univ Gustave Eiffel, F-77454 Marne-la-Vall\'ee, France.  Email: \texttt{eric.colin-de-verdiere}\texttt{@univ-eiffel.fr}}%
\and
Thomas Magnard\thanks{Former affiliation: LIGM, CNRS, Univ Gustave Eiffel, F-77454 Marne-la-Vall\'ee, France.  Email: \texttt{thomas.magnard}\texttt{@ac-creteil.fr}}%
}

\begin{document}
\maketitle

\begin{abstract}
  We consider the embeddability problem of a graph~$G$ into a two-dimensional simplicial complex~$C$:  Given $G$ and~$C$, decide whether $G$ admits a topological embedding into~$C$.  The problem is NP-hard, even in the restricted case where $C$ is homeomorphic to a surface.  We prove that the problem is fixed-parameter tractable in the size of the two-dimensional complex, by providing an $O(2^{\poly(c)}\cdot n^2)$-time algorithm.  If $G$ embeds into~$C$, we can compute a representation of an embedding in the same amount of time.  Moreover, we show that several known problems reduce to this one, such as the crossing number and the planarity number problems, and, under some conditions, the embedding extension problem.

  Our approach is to reduce to the case where $G$ has bounded branchwidth via an irrelevant vertex method, and to apply dynamic programming.  We do not rely on any component of the existing linear-time algorithms for embedding graphs on a fixed surface, but only on algorithms from graph minor theory.  However, by combining our results with a linear-time algorithm for embedding graphs on surfaces and with a very recent result for the irrelevant vertex method, we can decide whether $G$ embeds into~$C$ in $f(c)\cdot O(n)$ time, for some function~$f$.
\end{abstract}



\section{Introduction}\label{S:intro}

An embedding of a graph~$G$ into a host topological space~$X$ is a crossing-free topological drawing of~$G$ into~$X$.   The use and computation of graph embeddings is central in the communities of computational geometry/topology, algorithms for planar and surface-embedded graphs, topological graph theory, and graph drawing.  A landmark result is the algorithm of Hopcroft and Tarjan~\cite{ht-ept-74}, which allows to decide whether a given graph is planar (has an embedding into the plane) in linear time.  Related results include more planarity testing algorithms~\cite{p-pte-06}, algorithms for embedding graphs on surfaces~\cite{m-ltaeg-99,kmr-sltae-08} and for computing book embeddings~\cite{m-ggghp-94}, Hanani-Tutte theorems~\cite{s-ttpht-13}, and the theory of crossing numbers and planarization~\cite{s-gcnvs-13,bcgjm-cp-06}.

\paragraph{Our result.}
In this paper, we describe an algorithm for deciding the embeddability of graphs into topological spaces that are, in a sense, as general as possible: two-dimensional simplicial complexes (or 2-complexes for brevity), which are made from vertices, edges, and triangles glued together.  (We remark that every graph is embeddable in~$\mathbb R^3$ and thus in a 3-simplex, so considering higher-dimensional simplicial complexes is irrelevant.)  In a previous article, jointly written with Mohar~\cite{cmm-egtds-22}, we proved that, given a graph~$G$ and a 2-complex~$\complex$, deciding whether $G$~embeds into~$\complex$ is NP-complete and is in XP: it can be done in polynomial time for fixed~$\complex$, but the algorithm has running time $f(c)\cdot n^{O(c)}$, where $n$ and~$c$ are the respective sizes of $G$ and~$\complex$.  Using a very different strategy, we prove in this paper that it is actually fixed-parameter tractable (FPT) in the complexity of the input complex, by providing an algorithm that is quadratic in~$n$ and exponential in a polynomial in~$c$:
\begin{theorem}\label{T:general}
  One can solve the embeddability problem of graphs into 2-dimensional simplicial complexes in time $O(2^{\poly(c)}\cdot n^2)$, where $c$ is the number of simplices of the input 2-complex and $n$ is the total number of vertices and edges of the input graph.
\end{theorem}
(In our proof, the polynomial in~$c$ has degree $18+\varepsilon$, for any $\varepsilon>0$.)  We can also compute a representation of an embedding if there is one; see Section~\ref{S:computing} for details.

We remark that with additional ingredients, namely, a very recent paper~\cite{gkst-fivlt-25} and an algorithm to embed graphs on surfaces~\cite{m-ltaeg-99}, we can obtain an alternative decision algorithm that runs in time $f(c)\cdot O(n)$, for some unspecified function~$f$.  However, we only discuss this in the conclusion, Section~\ref{S:conclusion}.  Indeed, as we shall see, one of the benefits of our algorithm for Theorem~\ref{T:general} is that it does not rely on any of the (complicated) algorithms for embedding graphs on surfaces, and actually provides a simpler one, at the cost of a worse dependency in the input graph size~$n$.

\paragraph{Applications.}
Our embeddability problem is quite general and captures some known parameterized problems related to graph drawing and planarity, which are NP-hard in general, and we solve them in quadratic time, in a unified way.

The most relevant problem in this vein is the crossing number problem, namely, to decide whether a graph can be drawn in the plane with at most $k$ crossings.  The problem is NP-complete~\cite{gj-cnnc-83}.  As is already known~\cite[Introduction]{cmm-egtds-22}, there is an easy reduction from the crossing number problem to the embeddability problem of graphs in 2-complexes, which immediately implies a $2^{\poly(k)}\cdot O(n^2)$-time algorithm for the crossing number problem.  Previously, no FPT algorithm with runtime singly-exponential in~$k$ was known, although since the conference version of this paper, Lokshtanov, Panolan, Saurabh, Sharma, Xue, and Zehavi, in a very recent breakthrough paper~\cite{lpssxz-cnsst-25}, have given a $k^{O(k)}\cdot n$-time algorithm for the crossing number problem.

In a similar manner, the planarity number problem (whether a graph with a given set~$U$ of vertices can be embedded in the plane in such a way that $U$ is covered by at most $k$ faces) has an easy reduction to our problem, implying that it is FPT in~$k$, although the latter fact was proved earlier using very different techniques~\cite{bm-ccvfp-88}.

Actually, our techniques extend immediately to new, more general versions of all of these problems in which the plane is replaced with a fixed surface of genus~$g$; in that case, the algorithms are FPT in $k+g$.  We can also solve embedding extension problems, namely, embedding problems in which part of the graph is pre-embedded~\cite{m-ltaeg-99,angelini2015testing}.

All these results are detailed and proved in Section~\ref{S:applications}.  We also remark that, since the conference version of this paper appeared, more applications of Theorem~\ref{T:general} have been found; as it turns out, many kinds of crossing number problems encountered in graph drawing reduce to the embeddability of graphs on 2-complexes, so that our result applies to several other well-studied crossing number variants, for which fixed-parameter tractability parameterized by the total number of crossings was not known~\cite{ch-uffcn-25}.  Thus, our algorithm appears to be a versatile and powerful tool that allows to solve diverse problems in a relatively simple way.

\paragraph*{Embedding graphs on surfaces.}

Our embeddability problem also contains, as a special case, the problem of embedding a graph into a surface, since every surface of genus~$g$ (a ``sphere with $g$ handles or $g$ crosscaps'') is (homeomorphic to) some 2-complex with size $O(g)$.  This problem is of paramount importance; indeed, many graph problems turn out to be more easily solvable for graphs embedded on low-genus surfaces than for general graphs (see, e.g., the survey~\cite[Section~6]{c-ctgs-17}), so an embedding algorithm yields faster graph algorithms for abstract graphs that happen to be embeddable on low-genus surfaces.  The embedding problem is already NP-hard for surfaces~\cite{t-ggpnc-89}, and the existing algorithms that are fixed-parameter tractable in the genus are notoriously complicated; we review them now.

Mohar~\cite{m-ltaeg-99} has given an algorithm for embedding graphs on a fixed surface that takes linear time in the input graph, for every fixed surface.  This algorithm is very technical and relies on several other articles.  The dependence on the genus is not made explicit, but seems to be doubly exponential~\cite{kmr-sltae-08}.

Kawarabayashi, Mohar, and Reed, in an extended abstract~\cite{kmr-sltae-08}, have given a simpler linear-time algorithm for this problem, with a singly-exponential dependence in the genus, but not all details are presented, which makes the approach hard to check~\cite[p.~3657, footnote]{kp-dvgbg-19}.

General graph minor theory provides an algorithm for the same purpose.  The graph minor theorem by Robertson and Seymour~\cite{rs-gm20w-04} implies that, for every fixed surface~$\surf$, there is a finite list of graphs~$\cal O_\surf$ such that a graph~$G$ can be embedded on~$\surf$ if and only if $G$ does not contain any graph in~$\cal O_\surf$ as a minor.  Moreover, there is an algorithm that given any surface~$\surf$ (specified by its genus and orientability) outputs the list $\cal O_\surf$~\cite{agk-cem-08}.  Also, there is an algorithm to decide whether a fixed graph is a minor of an input graph~$G$ running in $O(n^{1+o(1)})$ time, where $n$ is the size of~$G$~\cite{kps-mcdpa-24} (improving over cubic-time algorithms~\cite{robertson1995graph}\cite[Theorem~6.12]{cfklmpps-pa-15}).  These considerations thus lead to an algorithm to decide embeddability of a graph on a surface that runs, if the input surface is fixed, in almost-linear time in the size of the input graph.  However, with this approach, it is not possible to compute an embedding if one exists.

Finally, in the same vein, Kociumaka and Ma.~Pilipczuk~\cite{kp-dvgbg-19} have studied the following problem, which is more general than the embeddability problem of graphs on surfaces:  Given a surface~$\surf$, a graph~$G$, and an integer~$k\ge0$, is it possible to remove a set~$W$ of at most~$k$ vertices from~$G$ so that $G-W$ is embeddable on~$\surf$?  They provide an algorithm that is fixed-parameter tractable in $k$ and the genus of~$\surf$, where the dependence on the genus is unspecified.  However, they use one of the previous algorithms for embedding graphs on surfaces~\cite{m-ltaeg-99,kmr-sltae-08} as a subroutine.  The problem that we study, the embeddability of graphs on 2-complexes, is independent from the problem studied by Kociumaka and Pilipczuk, in the sense that there is, a priori, no obvious reduction from one problem to the other.  However, we will reuse some ideas from that paper.

Our algorithm, restricted to the case where we want to embed graphs on surfaces, is not as efficient as the existing algorithms mentioned above, because it runs in quadratic time for fixed genus.  However, we stress that it does not rely on the toolbox used in the previous algorithms for embedding graphs on surfaces; we essentially rely on algorithms from graph minor theory (see in particular Lemma~\ref{L:branch-or-grid}) and on graph planarity testing.

\paragraph{Comparison between surfaces and 2-complexes.}
Every surface is homeomorphic to a 2-complex, but 2-complexes are much more general than surfaces, and tools that are suitable for studying embeddability of graphs on surfaces do not generalize.  For example, the set of graphs embeddable on a given 2-complex is not closed under minor, which makes many tools for dealing with graphs on surfaces unsuitable for 2-complexes.  Moreover, the complexity of some topological problems increase drastically when we consider 2-complexes instead of surfaces, e.g., deciding homeomorphism (solvable in linear time for surfaces, equivalent to graph isomorphism for 2-complexes~\cite{oww-h2ceg-00}), or deciding the contractibility of curves (solvable in linear time for surfaces~\cite{dg-tcs-99,lr-hts-12,ew-tcsr-13}, undecidable for 2-complexes~\cite{b-wp-59,s-ctcgt-93}).  See our previous paper with Mohar~\cite[Introduction]{cmm-egtds-22} for more details.  In this sense, it is perhaps surprising to obtain an FPT algorithm for embedding graphs on 2-complexes.  We also remark that the class of 2-complexes includes the \emph{pseudosurfaces} considered in graph theory; see Archdeacon~\cite[Section~5.7]{a-tgts-96} for a survey.

\paragraph*{Overview.}

At a very high level, we use a standard strategy in graph algorithms and parameterized complexity (see, e.g., the book by Cygan, Fomin, Kowalik, Lokshtanov, Marx, Pilipczuk, Pilipczuk, and Saurabh~\cite[Chapter~7]{cfklmpps-pa-15}): we show by dynamic programming that the problem can be solved efficiently for graphs of bounded branchwidth, and then, using an irrelevant vertex method, we prove that one can assume without loss of generality that the input graph~$G$ has branchwidth bounded by a polynomial in the size of the input 2-complex.  In the context of surface-embedded graphs, this paradigm has been used in the extended abstract by Kawarabayashi, Mohar, and Reed~\cite{kmr-sltae-08} and in the article by Kociumaka and Pilipczuk~\cite{kp-dvgbg-19}; our algorithm takes inspiration from the former one, for the idea of the dynamic programming algorithm, and from the latter, for some arguments in the irrelevant vertex method.  However, handling 2-complexes requires significantly more effort, new data structures, and some tailored preprocessing steps.  More precisely, Theorem~\ref{T:general} follows immediately from the following two theorems:
\begin{theorem}[algorithm for bounded branchwidth]\label{T:bounded-bw}
  We can solve the embeddability problem of graphs into two-dimensional simplicial complexes in time $(c+w)^{O(c+w)}\cdot n$, where $c$ is the number of simplices of the input 2-complex, $n$ is the total number of vertices and edges of the input graph, and $w$ is its branchwidth.
\end{theorem}
\begin{theorem}[algorithm to reduce branchwidth]\label{T:irrelevant}
  Given a 2-complex~$\complex$ with $c$ simplices, and a graph~$G$ with $n$ vertices and edges in total, we can, in time $O(2^{\poly(c)}\cdot n^2)$, do one of the following:
  \begin{itemize}
      \item correctly report that $G$ is not embeddable on~$\complex$,
      \item or compute a subgraph $G'$ of~$G$, of branchwidth polynomial in~$c$, such that $G$ embeds on~$\complex$ if and only if $G'$ does.
  \end{itemize}
\end{theorem}
(In our proof, the polynomial in~$c$ has degree~$18+\varepsilon$, for any $\varepsilon>0$.)

Incidentally, Theorem~\ref{T:bounded-bw} (together with the reduction presented in Section~\ref{S:crossing-number}) immediately implies that the crossing number~$k$ of a graph of branchwidth (or, equivalently, treewidth) $w$ can be computed in $(k+w)^{O(k+w)}\cdot n$ time, and this improves over the best previous result~\cite{kr-ccnlt-07}, in which the dependence on $k+w$ is triply exponential, since it relies on Courcelle's theorem~\cite{c-grala-90}.  As mentioned above, in a very recent paper~\cite{lpssxz-cnsst-25}, Lokshtanov, Panolan, Saurabh, Sharma, Xue, and Zehavi provide an algorithm for computing the crossing number~$k$ of arbitrary graphs in $k^{O(k)}\cdot n$ time.  To obtain this landmark result, they critically rely on an algorithm for the crossing number problem that runs in $(k+w)^{O(k+w)}\cdot n$ time.  While, for technical reasons in order to be reusable for the algorithm for arbitrary branchwidth, their bounded-branchwidth algorithm solves a more general problem than the crossing number problem, the present paper is the first to present such an algorithm.

\paragraph*{Structure of the paper.}

We now present the structure of the paper, indicating which techniques are used.  We also emphasize which components would be simpler if we were just aiming for an algorithm for embedding graphs on surfaces.

We introduce some standard notions in Section~\ref{S:prelim}.  

Sections \ref{S:2complexes} to~\ref{S:bounded-bw} contain the proof of Theorem~\ref{T:bounded-bw}.  In Section~\ref{S:2complexes}, we show that we can make some simplifying assumptions on the input.  In particular, every graph embeds in a 3-book (a 2-complex obtained by identifying three triangles along a common edge), so in the rest of the paper we assume that the input 2-complex has no 3-book.  We then present data structures for representing 2-complexes and graphs embedded on them.  If we restrict ourselves to the case where the input 2-complex is homeomorphic to a surface, we essentially consider combinatorial maps of graphs on surfaces, except that the graphs need not be cellularly embedded (such a data structure is called an \emph{extended combinatorial map}~\cite[Section~2.2]{cm-tgis-14}).  The case of 2-complexes is largely more involved.

In Section~\ref{S:partitioning}, we show that if our input graph~$G$ has an embedding into our input 2-complex~$\complex$, then there exists an embedding of~$G$ on~$\complex$ that is \emph{sparse} with respect to a branch decomposition of~$G$.  This means that each subgraph of~$G$ induced by the leaves of any subtree of the branch decomposition can be separated from the rest of~$G$ using a graph embedded on~$\complex$, called \emph{partitioning graph}, of small complexity.  We find that this new structural result, even in the surface case, is interesting and might prove useful in other contexts.  If the target space were a surface, we could assume that $G$ is 3-connected and has no loop or multiple edge, which would imply (still with some work) that \emph{any} embedding of~$G$ would be sparse, but again the fact that we consider 2-complexes requires additional work.

In Section~\ref{S:algo}, we present the dynamic programming algorithm, which either determines the existence of an embedding of~$G$ on~$\complex$, or shows that no sparse embedding of~$G$ on~$\complex$ exists (and thus no embedding at all, by the previous paragraph).  The idea is to use bottom-up dynamic programming and to consider all regions of the 2-complex in which the subgraph of~$G$ (induced by a subtree of the branch decomposition) can be embedded.  The complexity depends exponentially on the branchwidth of~$G$.

The previous arguments, most notably in Section~\ref{S:partitioning}, implicitly assumed that, if $G$ has an embedding into~$\complex$, it has a \emph{proper and cellular embedding}, in particular, in which the faces are homeomorphic to disks.  In Section~\ref{S:cell}, we show that we can assume this property.  Essentially, we build all 2-complexes ``smaller'' than~$\complex$, such that $G$ embeds on~$\complex$ if and only if it embeds into (at least) one of these 2-complexes, and moreover if it is the case, it has an embedding into (at least) one of these 2-complexes that is proper and cellular.  If $\complex$ were an orientable surface, we would just consider the surfaces of lower genus; but here a more sophisticated approach is needed.

The above ingredients allow to prove Theorem~\ref{T:bounded-bw} (Section~\ref{S:bounded-bw}).

There remains to prove that one can assume that $G$ has polynomial branchwidth (Theorem~\ref{T:irrelevant}); the proof, almost independent from the above, is contained in Section~\ref{S:irrelevant}.  It uses an irrelevant vertex method: in a nutshell, if $G$ has large branchwidth, we can compute a subdivision of a large wall, and then (unless $G$ has large genus and is not embeddable on~$\complex$) compute a large planar part of~$G$ containing a large wall; the central vertex of this wall is irrelevant, in the sense that its removal does not affect the embeddability or non-embeddability of the graph into~$\complex$; iterating, we obtain a graph of branchwidth polynomial in the size of~$\complex$.

In Section~\ref{S:computing}, we explain how to compute an embedding of~$G$ into~$\complex$, if one exists.  We prove the promised applications in Section~\ref{S:applications}.  Finally, we conclude in Section~\ref{S:conclusion}.

\section{Preliminaries}\label{S:prelim}

\subsection{Graphs and branch decompositions}\label{S:branchdecomp}

In this paper, graphs may have loops and multiple edges unless noted otherwise.  Let $G$ be a graph; as usual, we denote by $V(G)$ and $E(G)$ the sets of vertices and edges of~$G$.  \emphdef{Dissolving} a degree-two vertex~$v$ means replacing $v$ and its incident edges $vw_1$ and~$vw_2$ with a new edge between $w_1$ and~$w_2$.

A \emphdef{(rooted) branch decomposition} of~$G$ is a rooted tree~$\BB$ in which:
\begin{itemize}
    \item the root has a single child,
    \item every non-root node is either an \emphdef{internal node}, in which case it has exactly two children, which are ordered, or a \emph{leaf}, in which case it has no child,
    \item every (non-root) leaf is labelled with an edge of~$G$, and this labelling induces a bijection.
\end{itemize}
The vertices and edges of~$\BB$ are called \emphdef{nodes} and \emphdef{arcs}, respectively.  Clearly, the number of nodes of~$\BB$ is linear in the number of edges of~$G$.

Each arc $\alpha$ of~$\BB$ splits the tree~$\BB$ into two subtrees $\BB_1$ and~$\BB_2$; if, for $i=1,2$, we denote by $E_i$ the set of labels appearing in~$\BB_i$, we see that $\alpha$ naturally induces a partition $(E_1,E_2)$ of the set of edges of~$G$ (if $\alpha$ is the arc incident to the root, then one part of the partition is empty).  The \emphdef{middle set} $\mmid(\alpha)$ associated with~$\alpha$ is the set of vertices of~$G$ which are the endpoints of at least one edge in~$E_1$ and at least one edge in~$E_2$.

The \emphdef{width} of~$\BB$ is the maximum size of a middle set associated to an arc of~$\BB$.  The \emphdef{branchwidth} of~$G$ is the minimum width of its (rooted) branch decompositions.

The usual definition of a branch decomposition is identical, except that the tree is unrooted (thus the leaves are in bijection with the edges of~$G$) and children are unordered.  Our definition turns out to be more convenient to use in the dynamic program.  The difference is cosmetic: From any usual branch decomposition, one can trivially obtain a rooted branch decomposition of the same width, by subdividing an arbitrary arc with a new node~$\nu$ and then connecting~$\nu$ to a new leaf node~$\rho$, which will serve as the root; the converse operation obviously transforms any rooted branch decomposition into a usual branch decomposition.  Since each usual branch decomposition corresponds to a rooted branch decomposition, and both have the same width, we henceforth only work with rooted branch decompositions.

\subsection{Surfaces}

We will assume some familiarity with surface topology; see, e.g., \cite{mt-gs-01,s-ctcgt-93,c-ctgs-17} for suitable introductions under various viewpoints.   We recall some basic definitions and properties.  A \emphdef{surface} (without boundary) $\surf$ is a compact, connected Hausdorff topological space in which every point has an open neighborhood homeomorphic to the open disk.  Up to homeomorphism, every surface~$\surf$ is obtained from a sphere by:
\begin{itemize}
    \item either removing $g/2$ open disks and attaching a handle (a torus with an open disk removed) to each resulting boundary component, where $g$ is an even, nonnegative integer called the \emph{(Euler)} \emphdef{genus} of~$\surf$; in this case, $\surf$ is \emphdef{orientable};
    \item or removing $g$ open disks and attaching a M\"obius band to each resulting boundary component, for a positive number~$g$ called the (Euler) \emphdef{genus} of~$\surf$; in this case, $\surf$ is \emphdef{non-orientable}.
\end{itemize}
A \emphdef{possibly disconnected surface} is a disjoint union of surfaces.

A \emphdef{surface with boundary} is obtained from a surface (without boundary) by removing a finite set of interiors of disjoint closed disks.  The boundary of each disk forms a \emphdef{boundary component} of~$\surf$. The \emphdef{genus} of $\surf$ is defined as the genus of the original surface without boundary.    Equivalently, a surface with boundary is a compact, connected Hausdorff topological space in which every point has an open neighborhood homeomorphic to the open disk or the closed half disk $\{(x,y)\in\R^2\mid y\ge0, x^2+y^2<1\}$.

An embedding of a graph~$G$ into a surface~$\surf$ is \emphdef{cellular} if each face of the embedding is homeomorphic to an open disk.  If $G$ is cellularly embedded on a surface with genus~$g$ and $b$ boundary components, such that the embedding has $v$~vertices, $e$~edges, and $f$~faces, then Euler's formula stipulates that $v-e+f=2-g-b$ (these quantities are referred to as the \emphdef{Euler characteristic} of the surface).

\subsection{2-complexes}

A \emphdef{2-complex} (or two-dimensional simplicial complex) is an abstract simplicial complex of dimension at most two: a finite set of 0-simplices called \emphdef{vertices}, 1-simplices called \emphdef{edges}, and 2-simplices called \emphdef{triangles}.  Each edge is a pair of vertices, and each triangle is a triple of vertices; moreover, each subset of size two in a triangle must be an edge.  The 0-, 1-, and 2-simplices form the \emph{simplices} of the complex.  (We slightly depart from the standard definition in the sense that, for us, a 2-complex is a simplicial complex in which each simplex has dimension at most two; there needs not be a simplex of dimension exactly two.)

Each 2-complex~$\complex$ corresponds naturally to a topological space, obtained as follows: Start with one point per vertex in~$\complex$; connect them by segments as indicated by the edges in~$\complex$; similarly, for every triangle in~$\complex$, create a triangle whose boundary is made of the three edges contained in that triangle.  By abuse of language, we identify $\complex$ with that topological space.

\subsection{Graph embeddings}

Each graph has a natural associated topological space (for graphs without loops or multiple edges, this is a specialization of the definition for 2-complexes).  An \emphdef{embedding}~$\Gamma$ of a graph~$G$ into a 2-complex~$\complex$ is an injective continuous map from (the topological space associated to) $G$ to (the topological space associated to)~$\complex$.  A \emphdef{face} of~$\Gamma$ is a connected component of the complement of the image of~$\Gamma$ in~$\complex$.

\section{2-complexes and their data structures}\label{S:2complexes}

\subsection{Some preprocessing}

A \emphdef{3-book} is a topological space obtained from three triangles by considering one side per triangle and identifying these three sides together into a single edge.  We say that a 2-complex~$\complex$ \emphdef{contains a 3-book} if $\complex$ contains three distinct triangles that share a common edge.

\begin{proposition}\label{P:preproc}
  To decide the embeddability of a graph~$G$ on a 2-complex~$\complex$, we can without loss of generality, after a linear-time preprocessing, assume the following properties on the input:
  \begin{itemize}
      \item $\complex$ has no 3-book and no connected component that is reduced to a single vertex;
      \item $G$ has no connected component reduced to a single vertex, and at most one connected component that is a path.
  \end{itemize}
\end{proposition}
\begin{proof}
  It is known that every graph can be embedded into a 3-book~\cite[Proposition~3.1]{cmm-egtds-22}.  So we can without loss of generality assume that $\complex$ contains no 3-book.  We remove all the isolated vertices of~$\complex$, and remove the same number of isolated vertices of~$G$ (to the extent possible); this does not affect whether $G$ embeds into~$\complex$.  We then replace each isolated vertex of~$G$ with an isolated edge; since $\complex$ has no more isolated vertices, this does not affect embeddability of~$G$ into~$\complex$.  Finally, for the same reason, if $G$ contains at least two connected components that are paths, we replace all these components with a single edge.
\end{proof}

\textbf{In the rest of this article, without loss of generality, we implicitly assume that $\complex$ and~$G$ satisfy the properties stated in Proposition~\ref{P:preproc}.}

\subsection{Structure of 2-complexes without 3-book or isolated vertex}

\begin{figure}
\centering
\includegraphics[width=.9\linewidth]{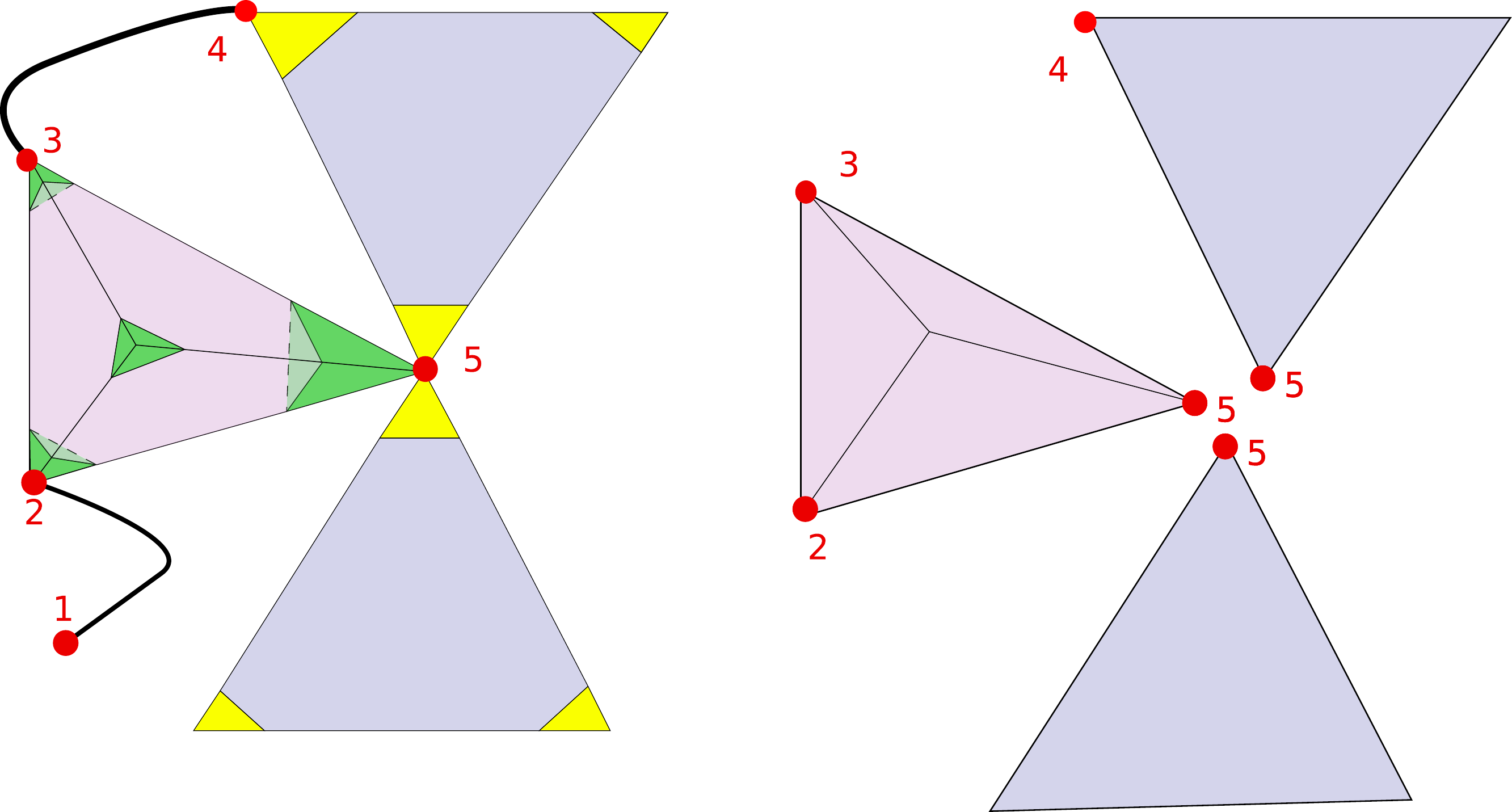}
\caption{On the left: A 2-complex with 5 singular vertices, numbered from 1 to~5, and 2 isolated edges (one between 3 and 4 and one between 1 and 2)  where, at singular vertices, the cones are in green and the corners in yellow. On the right: the corresponding detached surface.}
\label{F:Detached}
\end{figure}

Let $\complex$ be a 2-complex without 3-book or isolated vertex, and let $p$ be a vertex of~$\complex$.  Following~\cite[Section~2.2]{cmm-egtds-22}, we describe the possible neighborhoods of~$p$ in~$\complex$.  A \emphdef{cone at~$p$} is a cyclic sequence of triangles $t_1,\ldots,t_k,t_1$ ($k\ge3$), all incident to~$p$, such that, for each $i=1,\ldots,k$, the triangles $t_i$ and~$t_{i+1}$ (where $t_{k+1}=t_1$) share an edge incident with~$p$, and any other pair of these triangles have only $p$ in common.  A \emphdef{corner at~$p$} is an inclusionwise maximal sequence of distinct triangles $t_1,\ldots,t_k$, all incident to~$p$, such that, for each~$i=1,\ldots,k-1$, the triangles $t_i$ and~$t_{i+1}$ share an edge incident with~$p$, any other pair of these triangles have only $p$ in common.  An \emphdef{isolated edge at~$p$} is an edge incident to~$p$ but not incident to any triangle.  The cones, corners, and isolated edges at~$p$ form the \emphdef{link components} at~$p$.

The set of edges and triangles incident with a given vertex~$p$ of~$\complex$ are uniquely partitioned into cones, corners, and isolated edges.  We say that $p$ is a \emphdef{regular vertex} if all the edges and triangles incident to~$p$ form a single cone or corner; in that case, $p$ has an open neighborhood homeomorphic to a disk or a closed half-disk.  Otherwise, $p$ is a \emphdef{singular vertex}.  See Figure~\ref{F:Detached}, left, for an illustration.

\emphdef{Detaching} a singular vertex~$v$ in~$\complex$ consists of the following operation: replace $v$ with new vertices, one for each cone, corner, and isolated edge at~$v$.  Detaching all singular vertices of a 2-complex (without 3-book) yields a disjoint union of (1) isolated edges and (2) a surface, possibly disconnected, possibly with boundary, called the \emphdef{detached surface} (see Figure~\ref{F:Detached}, right).  The trace of the singular vertices on the detached surface are the \emphdef{marked vertices}.  Conversely, $\complex$ can be obtained from a surface (possibly disconnected, possibly with boundary) and a finite set of segments by choosing finitely many subsets of points and identifying the points in each subset together.

The \emphdef{boundary} of~$\complex$ is the closure of the set of points of~$\complex$ that have an open neighborhood homeomorphic to a closed half-plane.  Equivalently, it is the union of the edges of~$\complex$ incident with a single triangle.

\subsection{Topological data structure for 2-complexes}\label{S:datastruct-2complexes}

We now describe a \emphdef{topological data structure} for 2-complexes without 3-book or isolated vertex that is more appropriate for our purposes.  It records only the 2-complex up to homeomorphism, not the combinatorial information given by its simplices.  Such a 2-complex~$\complex$ is obtained from a surface~$\surf$ (possibly disconnected, possibly with boundary) and a finite set~$S$ of segments by identifying together finitely many finite subsets of points.  Our data structure stores separately the detached surface~$\surf$, the set~$S$ of isolated edges, and the singular vertices, and two-way pointers representing incidences between them.  In more detail:
\begin{itemize}
    \item we store the list of the connected components of the detached surface~$\surf$, and for each such component~$\surf'$ we store (1) whether it is orientable or not; (2) its genus; (3) a list of pointers to the singular vertices in the interior of~$\surf'$; (4) for each boundary component of~$\surf'$, a cyclically ordered list of pointers to the singular vertices appearing on that boundary component (if $\surf'$ is orientable, the boundary components must be traversed in an order consistent with an arbitrarily chosen orientation of~$\surf'$);
    \item we store the list~$S$ of isolated edges, and for each of them, two pointers to the singular vertices at its endpoints;
    \item conversely, to each singular vertex is attached a list of pointers to the occurrences of that singular vertex on the detached surface or as an endpoint of an isolated edge.
\end{itemize}

The \emphdef{size} of a 2-complex (without 3-book or isolated vertex) is the sum of the number of isolated edges, the number of connected components of the detached surface, the total genus of the detached surface, the total number of boundary components of the detached surface, and the total number of marked vertices of the detached surface (the occurrences of the singular vertices).  This is, up to a constant factor, the size of the topological data structure indicated above, if the genus is stored in unary.

Given a 2-complex~$\complex$ without 3-book or isolated vertex, described by vertices, edges, and triangles and their incidence relations, we can easily compute a representation of~$\complex$ in that data structure, in polynomial time in the number of simplices: Indeed, by ignoring the incidences created by vertices, we easily build a triangulation of the surface~$\surf$ (possibly disconnected, possibly with boundary) and a list of segments~$S$; we then compute the topology of~$\surf$; finally, we mark the singular vertices, which are the vertices with several occurrences on~$\surf$ and/or on~$S$.  We remark that the size of the resulting data structure is at most linear in the number of simplices of the 2-complex~$\complex$, because, by Euler's formula, any triangulated surface (possibly with boundary) with $k$ simplices has genus $O(k)$ and a number of boundary components that is $O(k)$.  Thus, \textbf{in the rest of this article, without loss of generality, we implicitly assume that $\complex$ is given in the form of the above topological data structure.}  (Conversely, it is not hard to see that every 2-complex is homeomorphic to a 2-complex whose number of simplices that is linear in its size, but we will not need this fact.)

We will need the following lemma.
\begin{lemma}\label{L:homeo-2complex}
  Given two 2-complexes $\complex$ and~$\complex'$, given in the topological representation above, of sizes $c$ and~$c'$ respectively, we can decide whether $\complex$ and~$\complex'$ are homeomorphic in time $(c+c')^{O(c+c')}$.
\end{lemma}
\begin{proof}
  (We remark that this essentially follows from more general results~\cite{oww-h2ceg-00}; the running time of our algorithm might be improvable, but this suffices for our purposes.)  As a preprocessing, in the topological data structures of~$\complex$ and~$\complex'$, we do the following: whenever a singular vertex is incident to exactly two isolated edges and is not incident to the detached surface, we dissolve that singular vertex, removing it and replacing the two incident edges with a single one.  Clearly, this does not affect whether $\complex$ and~$\complex'$ are homeomorphic.
  
  After this preprocessing, $\complex$ and~$\complex'$ are homeomorphic if and only if their topological data structures are isomorphic.  By this, we mean that there is a bijective correspondence~$\varphi$ from the isolated edges, the connected components of the detached surface, and the boundary components of each connected component of the detached surface of~$\complex$ to those of~$\complex'$ that preserves the genus, the orientability, the incidences, and the cyclic ordering of the singular vertices on each boundary component.  More precisely, for the latter point: for each connected component~$C$ of the detached surface of~$\complex$, if $C$ is orientable, then the lists of singular vertices appearing on each boundary component of~$C$ and~$\varphi(C)$ are identical up to global reversal of all these cyclic orderings simultaneously, corresponding to a change of the orientation of the connected component; if $C$ is non-orientable, then the lists of singular vertices appearing on each boundary component of~$C$ and~$\varphi(C)$ are identical up to the possible individual reversal of some of these cyclic orderings.  The proof is tedious but straightforward, and the existence of an isomorphism can obviously be tested in the indicated time.
\end{proof}

\subsection{Proper and cellular graph embeddings on 2-complexes}\label{S:datastruct-graphs}

Let $\complex$ be a 2-complex with size~$c$, $G$ a graph, and $\Gamma$ an embedding of~$G$ on~$\complex$.  The embedding~$\Gamma$ is \emphdef{proper} if:
\begin{itemize}
    \item the image of~$\Gamma$ meets the boundary of~$\complex$ only on singular vertices;
    \item the vertices of~$\Gamma$ cover the singular vertices of~$\complex$.
\end{itemize}

The embedding~$\Gamma$ is \emphdef{cellular} if each face of~$\Gamma$ is an open disk plus possibly some part of the boundary of~$\complex$.  We emphasize that this definition slightly departs from the standard one.  Moreover, we will only consider cellular embeddings that are proper.

Traditional data structures for graphs on surfaces handle graphs embedded cellularly; \emph{rotation systems}~\cite{mt-gs-01} constitute one example of such a data structure.  In order to have efficient algorithms, refined data structures, e.g., with graph-encoded maps~\cite{l-gem-82} (see also~\cite[Section~2]{e-dgteg-03}), are needed.  The basic element in the graph-encoded map is the \emph{flag}, an incidence between a vertex, an edge, and a face of the graph.  Three involutions allow to move from each flag to a nearby flag.  Each flag contains a pointer to the underlying vertex, edge, and face.

One can easily extend such data structures to possibly non-cellular embeddings on surfaces~\cite[Section~2.2]{cm-tgis-14}.  In this framework, one must store the topology of each face, which is not necessarily homeomorphic to a disk.  Also, a face may have several boundary components; two-way pointers connect each face to one flag of each boundary component (or to an isolated vertex of the graph, if that boundary component is a single vertex); if a face is orientable and has several boundary components, then these pointers must induce a consistent orientation of these boundary cycles.  It is important to remark that this data structure also allows to recover the topology of the underlying surface.

Let $\Gamma$ be a proper graph embedding of a graph~$G$ on a 2-complex~$\complex$ (under the assumptions of Proposition~\ref{P:preproc}).  Let $\surf$ be the detached surface of~$\complex$.  Because $\Gamma$ is proper, it naturally induces an embedding~$\Gamma'$, of another graph~$G'$, on~$\surf$; each vertex of~$G$ located on a singular vertex of~$\complex$ appears as many times in~$G'$ as there are cones and corners at that singular vertex; the vertices of~$G$ located in the relative interior of isolated edges are absent from~$G'$.  Our data structure, called \emphdef{combinatorial map}, for storing the graph embedding~$\Gamma$ and the 2-complex~$\complex$ consists of storing (1) the graph embedding~$\Gamma'$ on~$\surf$, as indicated in the previous paragraph, (2) the isolated edges of~$\complex$, together with, for each such isolated edge, an ordered list alternating vertices and edges of~$\Gamma$ (or, instead of an edge, a mark indicating the absence of such an edge in the region of the isolated edge between the incident vertices), (3) the identifications of vertices of~$\Gamma'$ that are needed to recover~$\Gamma$ (and thus implicitly~$\complex$).

Isomorphisms between combinatorial maps are defined in the obvious way, similar to the concept of isomorphism between topological data structures:  Two combinatorial maps are isomorphic if there is an isomorphism between the combinatorial maps restricted to the detached surfaces, isomorphisms between the maps on each isolated edges, and such that incidences are preserved on the singular vertices.  We can easily test isomorphism between two combinatorial maps of size $k$ and~$k'$, respectively, in $(k+k')^{O(k+k')}$ time.

We will need an algorithm to enumerate all proper embeddings of small graphs on a given 2-complex.  This is achieved in the following lemma.
\begin{lemma}\label{L:enum-embeddings}
  Let $\complex$ be a 2-complex of size~$c$ and $k$ an integer.  We can enumerate the ${(c+k)}^{O(c+k)}$ combinatorial maps of graphs with at most $k$ vertices and at most $k$ edges properly embedded on~$\complex$ in ${(c+k)}^{O(c+k)}$ time.
\end{lemma}
\begin{proof}
  The strategy is the following.  In a first step, we enumerate a set of proper graph embeddings on some 2-complexes, which necessarily contains all the desired combinatorial maps.  In a second step, we prune this set to keep only the desired combinatorial maps, by eliminating those that contains too many vertices or edges, or that correspond to an embedding on a 2-complex not homeomorphic to~$\complex$.

  \medskip
 
  \emph{First step.}  Let $\Gamma$ be a proper embedding of a graph with at most $k$~vertices and $k$~edges on~$\complex$.  Let $\surf$ be the detached surface associated to~$\complex$; this surface is possibly disconnected and has genus at most~$c$.  The image of~$\Gamma$ on~$\surf$ is a graph with at most $k+c$~vertices and $k$~edges.
  
  We first enumerate, in a possibly redundant way, the set~$M_1$ of combinatorial maps of cellular graph embeddings with at most $k+c$~vertices and $k$~edges on a possibly disconnected surface without boundary.  There are $(c+k)^{O(c+k)}$ such combinatorial maps, which can be enumerated in $(c+k)^{O(c+k)}$ time, for example using rotation systems.  Some vertices may be isolated, if the corresponding connected component of~$\surf$ is a sphere.
  
  By simplifying the surface, every graph embedding can be transformed into a cellular graph embedding: remove each face and paste a disk to each cycle of the graph that was a boundary component of a face.  Conversely, every (possibly non-cellular) graph is obtained from some cellular one by (1) connecting some faces together (creating a face of genus zero with several boundary components) and (2) adding some genus and non-orientability to some faces.  So, for each map in~$M_1$, we perform all these operations in all possible ways, by putting genus at most~$c$ in each face.  For each such map in~$M_1$, there are $(c+k)^{O(c+k)}$ possibilities.  We thus obtain, in $(c+k)^{O(c+k)}$ time, a set $M_2$ of $(c+k)^{O(c+k)}$ combinatorial maps on surfaces, and the set~$M_2$ contains all combinatorial maps of (possibly non-cellular) graph embeddings on surfaces of genus at most~$c$.
  
  Finally, we add at most~$c$ isolated edges, choose how endpoints of these isolated edges and vertices of the embedding on the detached surface are identified, and decide how each isolated edge is covered by the embedding. There are $(c+k)^{O(c+k)}$ ways to do this.  We thus have computed, in $(c+k)^{O(c+k)}$ time, a set~$M$ of $(c+k)^{O(c+k)}$ combinatorial maps of graphs on 2-complexes, which contains all the combinatorial maps indicated in the statement of the lemma.

  \emph{Second step.}  First, we easily discard the combinatorial maps in~$M$ containing more than $k$~vertices or $k$~edges.  Then, we discard the maps in~$M$ corresponding to a 2-complex different from~$\complex$.  For this purpose, for each map~$m$ in~$M$, we iteratively remove the edges of the graph embedding, preserving the underlying 2-complex.  When removing an edge from the detached surface, the topology of the incident face(s) change; we preserve this information.  Finally, we remove every isolated vertex that does not lie on a singular vertex of the 2-complex.  The data structure that we have now is essentially the one that is described in Section~\ref{S:datastruct-2complexes}; we can thus easily decide whether that 2-complex is homeomorphic to~$\complex$ (Lemma~\ref{L:homeo-2complex}), and discard~$m$ if and only if it is not the case.
  
  Finally, and although this is not strictly needed, we can easily remove the duplicates in the combinatorial maps, by testing pairwise isomorphism between these maps.
\end{proof}

\subsection{Graphs embeddable on a fixed 2-complex have bounded genus}

\begin{figure}
\centering
\includegraphics[width=.9\linewidth]{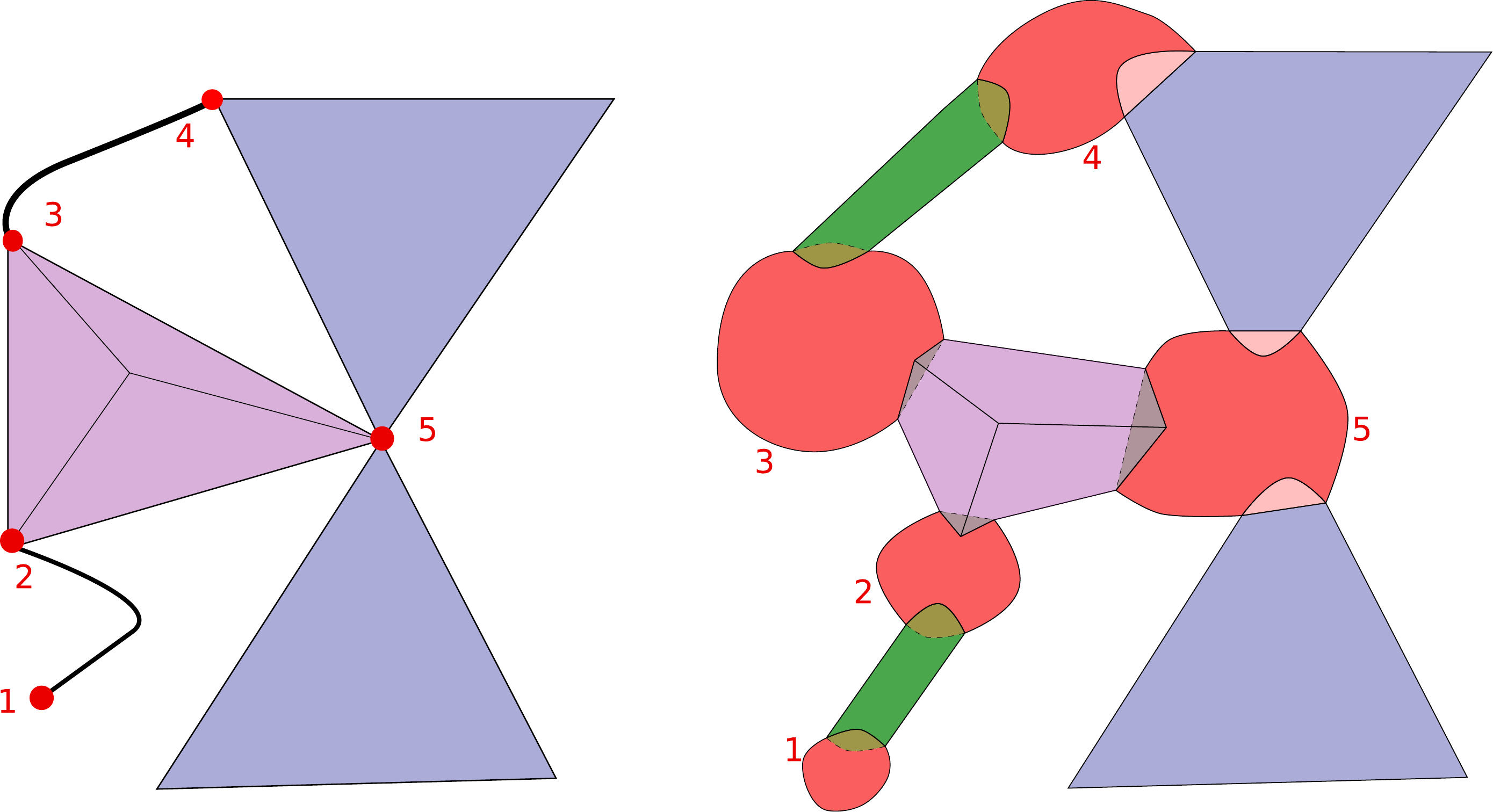}
\caption{On the left: The same 2-complex as Figure~\ref{F:Detached}. On the right: the corresponding surface constructed in Lemma \ref{L:genus-oversurface}.}
\label{F:Supersurface}
\end{figure}

\begin{lemma}\label{L:genus-oversurface}
  Let $\complex$ be a 2-complex without 3-book.  Let $c$ be either the size of~$\complex$ or its number of simplices.  Every graph embeddable on~$\complex$ is embeddable on a surface of (Euler) genus at most $10c$.
\end{lemma}
\begin{proof}
  The strategy is to construct a surface~$\surf$ of genus at most~$10c$ such that every graph embeddable on~$\complex$ is also embeddable on~$\surf$.  The surface~$\surf$ is obtained by replacing every isolated edge of~$\complex$ with a cylinder, and modifying the structure of the 2-complex in the neighborhood of each singular vertex to make it surface-like; see Figure~\ref{F:Supersurface}.
  
  In more detail, we can obviously assume that $\complex$ has no isolated vertex.  First, we replace every isolated edge of~$\complex$ with a cylinder.  Then, for every singular vertex~$v$, we do the following.  We remove a small neighborhood of~$v$.  We create a sphere with $k$~boundary components, where $k$ is the number of link components at~$v$.  Finally, we attach a link component to each of the $k$ boundary components of the sphere:
  \begin{itemize}
      \item for each link component that was a cone, a small neighborhood of~$v$ was removed, with boundary a circle; we attach that circle bijectively to a boundary component of the sphere;
      \item for each link component that was a corner, a small neighborhood of~$v$ was removed, with boundary a segment; we attach that segment to a part of a boundary component of the sphere;
      \item for each link component that was an isolated edge, the isolated edge was replaced with a cylinder; we attach the corresponding boundary component of that cylinder bijectively to a boundary component of the sphere.
  \end{itemize}
  
  (We remark that there are several ways to perform these operations, depending on the orientation of the gluings; any choice will do.)  The resulting surface~$\surf$, which is possibly non-connected and possibly with boundary, has genus at most~$10c$.  Indeed, if $c$ is the size of~$\complex$, this follows from Euler's formula (intuitively, the number of ``handles'' created is at most the number of link components).  If $c$ is the number of simplices of~$\complex$, it follows from the fact that the total genus of the detached surface is at most~$c$, again by Euler's formula, and from the fact that each isolated edge or triangle contributes to an increase of at most six for the genus in the construction above.
  
  Consider an embedding of a graph~$G$ on~$\complex$.  It is not hard to transform that embedding into an embedding of~$G$ on~$\surf$: Each cylinder replacing an isolated edge is used only along a single path connecting its two boundary components; if a singular vertex~$v$ is used by the embedding of~$G$ on~$\complex$, we can locally modify the embedding to accommodate the local change from~$\complex$ to~$\surf$ at~$v$ (see again Figure \ref{F:Supersurface}).  Of course, if $G$ is embeddable on~$\surf$, it is embeddable on some connected surface of genus at most~$10c$.
\end{proof}

As a side remark, it follows from the above lemma and from Euler's formula that, if $G$ is simple and embeds on~$\complex$, then its number $v$ of vertices and $e$ of edges satisfy  $e\le 3v+30c$.  We will actually not use this fact, in particular because we do not assume $G$ to be simple, and because, for simplicity, we chose to express all our results in terms of $n$, the total number of vertices and edges of~$G$, and not in terms of the number of vertices only.

\section{Partitioning graphs}\label{S:partitioning}

Let $\complex$ be a 2-complex and $G$ a graph, which satisfy the properties of Proposition~\ref{P:preproc}.  In this section, we lay the structural foundations of the dynamic programming algorithm, described in the next section (Proposition~\ref{P:algo}).  The goal, in this section and the following one, is to obtain an algorithm that takes as input $\complex$ and $G$ and, in time FPT in the size of~$\complex$ and the branchwidth of~$G$, reports correctly one of the following two statements:
\begin{itemize}
  \item $G$ has no proper cellular embedding on~$\complex$,
  \item $G$ has an embedding on~$\complex$,
\end{itemize}

This algorithm uses dynamic programming on a rooted branch decomposition of~$G$.  When processing a node of the rooted branch decomposition, it considers embeddings of the subgraph of~$G$ induced by the edges in the leaves of the subtree rooted at that node in a region of~$\complex$.  This region will be delimited by a \emph{partitioning graph} embedded on~$\complex$.  Our dynamic program will roughly guess the partitioning graph at each node of the rooted branch decomposition.  For this purpose, we need that, if $G$ has a proper cellular embedding on~$\complex$, it has such an embedding that is \emph{sparse}: at each node of the rooted branch decomposition of~$G$, the partitioning graph corresponding to the embedding of the induced subgraph is small (its size is upper-bounded by a function of the branchwidth of~$G$ and of the size of~$\complex$).  The goal of this section is to prove that this is indeed the case.

Let $(E_1,\ldots,E_k)$ be a partition of the edge set~$E(G)$ of~$G$.  (We will only use the cases $k=2$ and $k=3$.)  The \emphdef{middle set} $\mmid(E_1,\ldots,E_k)$ of $(E_1,\ldots,E_k)$ is the set of vertices of~$G$ whose incident edges belong to at least two sets~$E_i$.

Let $\Gamma$ be a proper cellular embedding of~$G$ on~$\complex$.  Since $\Gamma$ is cellular, every boundary of~$\complex$ is incident to at least one vertex of~$\Gamma$.  The \emphdef{partitioning graph}~$\Pi(\Gamma,E_1,\ldots,E_k)$ (or more concisely~$\Pi$) associated to~$\Gamma$ and $(E_1,\ldots,E_k)$ is a graph properly embedded on~$\complex$ (but possibly non-cellularly), with labels on its faces, defined as follows:
\begin{figure}
\centering
\footnotesize
\def\svgwidth{.7\linewidth}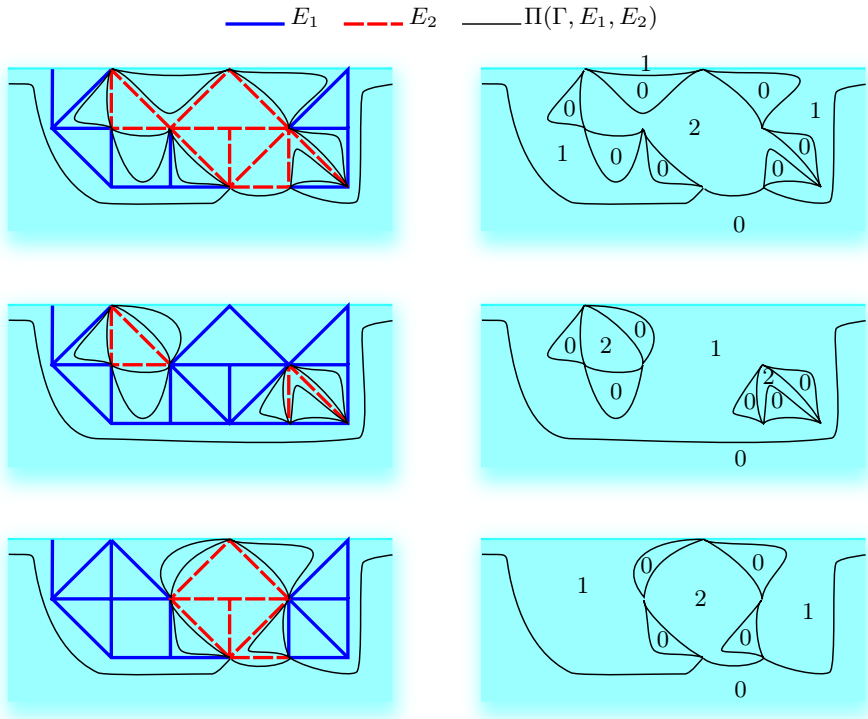
\caption{Construction of the partitioning graph $\Pi=\Pi(\Gamma,E_1,E_2)$, for three choices of the partition $(E_1,E_2)$ of the same embedding~$\Gamma$.  Only a part of the 2-complex~$\complex$ is shown, with a boundary at the upper part, and without singular vertex.  Left: The graph embeddings $\Gamma$ (in thick lines) and $\Pi$ (in thin lines).  Right: The sole graph~$\Pi$, together with the labelling of its faces.}
\label{F:partitioning}
\end{figure}
\begin{figure}
\centering
\footnotesize
\def\svgwidth{.7\linewidth}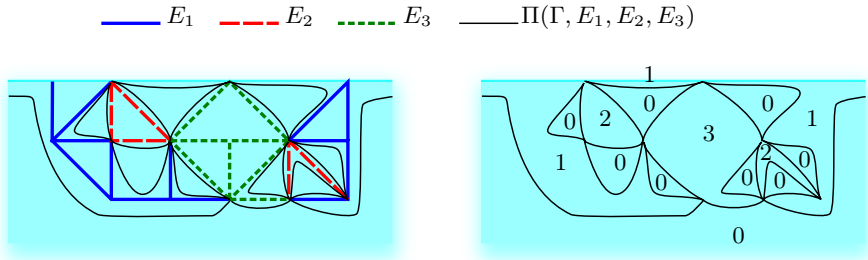
\caption{The partitioning graph $\Pi=\Pi(\Gamma,E_1,E_2,E_3)$.  Left: The graph embeddings $\Gamma$ (in thick lines) and $\Pi$ (in thin lines).  Right: The sole graph~$\Pi$, together with the labelling of its faces.}
\label{F:partitioning-new-node}
\end{figure}

\begin{itemize}
    \item The vertex set of~$\Pi$ is the union of the singular vertices of~$\complex$ and of (the images under~$\Gamma$ of) $\mmid(E_1,\ldots,E_k)$.  Each vertex of~$\Pi$ is labelled~$v\in\mmid(E_1,\ldots,E_k)$ if it is mapped, on~$\complex$, to the same location as~$v$ under~$\Gamma$; the other vertices of~$\Pi$ are unlabelled.
    
    \item The relative interiors of the edges of~$\Pi$ are disjoint from the edges of~$\Gamma$ and from the isolated edges of~$\complex$. Let $f$ be a face of~$\Gamma$; because $\Gamma$ is cellular, $f$ is homeomorphic to an open disk plus possibly some points of the boundary of~$\complex$. Let us describe the edges of~$\Pi$ inside~$f$.

    If, for some~$i\in\{1,\ldots,k\}$, the boundary of~$f$ is comprised only of edges of~$\Gamma$ that lie in a single set~$E_i$ (together with their endpoints), then $\Pi$ contains no edge inside~$f$.  Similarly, if the boundary of~$f$ is entirely included in the boundary of~$\complex$, then $\Pi$ contains no edge inside~$f$.
    
    Otherwise, the boundary of~$f$ is a succession of edges of $E_1$, $E_2$, \ldots, $E_k$, and of pieces of the boundary of~$\complex$ (possibly including singular vertices).  The edges of~$\Pi$ inside~$f$ run along the boundary of~$f$; for each~$i\in\{1,\ldots,k\}$, for each (maximal) group of consecutive edges in $E_i$ along the boundary of~$f$, there is an edge of~$\Pi$ that runs along this group, with endpoints the corresponding vertices on the boundary of~$f$ (see Figures~\ref{F:partitioning} and~\ref{F:partitioning-new-node}).  These vertices either are in $\mmid(E_1,\ldots,E_k)$, or lie on the boundary of~$\complex$ (and thus on singular vertices of~$\complex$, because $\Gamma$ is proper).   It follows from the construction that $\Gamma$ and~$\Pi$ intersect only at common vertices.
 
    \item Each face of~$\Pi$ is labelled by an integer in~$\{0,\ldots,k\}$ as follows: faces of~$\Pi$ containing an edge in~$E_i$ are labelled~$i$, and the other ones are labelled~$0$.  Since $\Pi$ contains the singular vertices of~$\complex$ and $\mmid(E_1,\ldots,E_k)$, and by the construction of the edges of~$\Pi$, this labelling is well defined.  We remark that this definition is valid also for faces of~$\Pi$ that are parts of an isolated edge of~$\complex$.
\end{itemize}

Henceforth, we fix a rooted branch decomposition~$\BB$ of~$G$.  Every arc~$\alpha$ of~$\BB$ naturally partitions~$E(G)$ into two parts $E_1$ and~$E_2$; this partition is the \emphdef{edge partition associated to~$\alpha$}.  Although this will play a role only in the next section, we choose $E_1$ and $E_2$ so that $E_1$ is on the side of the root of the branch decomposition~$\BB$.  Recall that $\Gamma$ is a proper and cellular embedding of~$G$ on~$\complex$; we let $\Pi(\Gamma,\alpha)$ be $\Pi(\Gamma,E_1,E_2)$.  Similarly, every internal node~$\nu$ of~$\BB$ naturally partitions~$E(G)$ into three parts $E_1$, $E_2$, and~$E_3$, in which $E_1$ is on the side of the edge incident with~$\nu$ that is the closest to the root, $E_2$ corresponds to the edges of~$G$ that are leaves of the subtree rooted at the first child of~$\nu$, and $E_3$ corresponds to the edges of~$G$ that are leaves of the subtree rooted at the second child of~$\nu$.  This partition is the \emphdef{edge partition associated to~$\nu$}; we let $\Pi(\Gamma,\nu)$ be $\Pi(\Gamma,E_1,E_2,E_3)$.

We say that $\Gamma$ is \emphdef{sparse} (with respect to~$\BB$) if the following conditions hold, letting $c$ be the size of~$\complex$ and $w$ the width of~$\BB$:
\begin{itemize}
    \item For each arc~$\alpha$ of~$\BB$, the graph $\Pi(\Gamma,\alpha)$ has at most $75c+27w$ edges;
    \item similarly, for each internal node~$\nu$ of~$\BB$, the graph $\Pi(\Gamma,\nu)$ has at most $3(75c+27w)/2$ edges.
\end{itemize}

The result of this section is the following.
\begin{proposition}\label{P:sparse}
  Let $\complex$ be a 2-complex and $G$ a graph, which satisfy the properties of Proposition~\ref{P:preproc}.  Let $\BB$ be a rooted branch decomposition of~$G$.  Assume that $G$~has a proper cellular embedding on~$\complex$.  Then it has a proper cellular embedding~$\Gamma$ on~$\complex$ that is sparse (with respect to~$\BB$).
\end{proposition}

\subsection{Monogons and bigons}

A \emphdef{monogon} of a graph~$\Pi$ embedded on a 2-complex~$\complex$ is a face of~$\Pi$ that is an open disk whose boundary is a single edge of~$\Pi$ (a loop).  Similarly, a \emphdef{bigon} of~$\Pi$ is a face of~$\Pi$ that is an open disk whose boundary is the concatenation of two edges of~$\Pi$ (possibly the same edge appearing twice).  The following general lemma on graphs embedded on surfaces without monogons or bigons will be used below, and also in Section~\ref{S:irrelevant}; some particular cases have been used before~\cite[Lemma~2.1]{ccelw-scsh-08}.
\begin{lemma}\label{L:nogon}
  Let $\surf$ be a surface of genus~$g$ without boundary.  Let $\Pi$ be a graph embedded on~$\surf$, not necessarily cellularly.  Assume that $\Pi$ has no monogon or bigon.  Then $|E(\Pi)|\le\max\{0, 3g+3|V(\Pi)|-6\}$.
\end{lemma}
\begin{proof}
  We begin by adding edges to~$\Pi$ as long as it is possible to do so, without introducing any new vertex, monogon, or bigon.  Let $\Pi'$ be the resulting embedded graph.  We claim that every face of~$\Pi'$ is a triangle (a disk incident with three edges), except in the following cases:
  \begin{enumerate}
      \item $\Pi'$ is the empty graph;
      \item $\surf$ is a sphere, and $\Pi'$ has two vertices and no edge;
      \item $\surf$ is a sphere, and $\Pi'$ has a single vertex and no edge;
      \item $\surf$ is a projective plane, and $\Pi'$ has a single vertex and no edge.
  \end{enumerate}
  Indeed, let $f$ be a face of~$\Pi'$.  If $f$ has no boundary component, then since $\surf$ is connected, we are in Case~1 (an isolated vertex would account for a boundary component).  Assume that $f$ has at least two boundary components.  We add an edge in~$f$ connecting vertices on these two boundary components.  This cannot create any monogon or bigon, except if the two boundary components are both reduced to a single vertex and $\surf$ is a sphere (Case~2).  So we can assume that $f$~has a single boundary component.  If $f$ is orientable and has genus zero, then either we are in Case~3, or $\surf$~is a disk of degree at least four, in which case we can add an edge to split it into smaller disks without creating any monogon or bigon.  If $f$ is orientable and has Euler genus at least two (i.e., orientable genus at least one), we can add an edge that forms a non-separating arc (relatively to the boundary) in~$f$; it does not form any monogon or bigon.  If $f$ is non-orientable and has (non-orientable) genus at least two, we can add an edge that forms a separating arc in~$f$, cutting that surface into two non-orientable surfaces of genus at least one; it does not form any monogon or bigon.  Finally, if $f$ is non-orientable and has (non-orientable) genus one, either the boundary component is reduced to a single vertex, so we are in Case~4, or this face has degree at least one; then it is a M\"obius band with at least one vertex and one edge on its boundary, and we can add a loop that is a non-contractible arc (relatively to the boundary) in~$f$, without forming any monogon or bigon.  The only remaining possibility is that $f$ is a triangle.
  
  It is clear that the statement of the lemma holds whenever we are in one of the four above cases.  So we can assume without loss of generality that each face of~$\Pi'$ is a triangle.  Since $V(\Pi)=V(\Pi')$ and $|E(\Pi)|\le|E(\Pi')|$, it suffices to prove the result for~$\Pi'$ instead of~$\Pi$.  Double-counting the incidences between edges and faces implies that the number of triangles~$\tau$ satisfies $3\tau=2|E(\Pi')|$; plugging this into Euler's formula implies that $|V(\Pi')|-|E(\Pi')|/3=2-g$, so $|E(\Pi')|=3g+3|V(\Pi')|-6$, as desired.
  \end{proof}

\subsection{Vertex simplifications}

\begin{figure}
  \centering
  \def\svgwidth{.7\linewidth}
\begingroup%
  \makeatletter%
  \providecommand\color[2][]{%
    \errmessage{(Inkscape) Color is used for the text in Inkscape, but the package 'color.sty' is not loaded}%
    \renewcommand\color[2][]{}%
  }%
  \providecommand\transparent[1]{%
    \errmessage{(Inkscape) Transparency is used (non-zero) for the text in Inkscape, but the package 'transparent.sty' is not loaded}%
    \renewcommand\transparent[1]{}%
  }%
  \providecommand\rotatebox[2]{#2}%
  \newcommand*\fsize{\dimexpr\f@size pt\relax}%
  \newcommand*\lineheight[1]{\fontsize{\fsize}{#1\fsize}\selectfont}%
  \ifx\svgwidth\undefined%
    \setlength{\unitlength}{346.42762559bp}%
    \ifx\svgscale\undefined%
      \relax%
    \else%
      \setlength{\unitlength}{\unitlength * \real{\svgscale}}%
    \fi%
  \else%
    \setlength{\unitlength}{\svgwidth}%
  \fi%
  \global\let\svgwidth\undefined%
  \global\let\svgscale\undefined%
  \makeatother%
  \begin{picture}(1,0.45582368)%
    \lineheight{1}%
    \setlength\tabcolsep{0pt}%
    \put(0,0){\includegraphics[width=\unitlength,page=1]{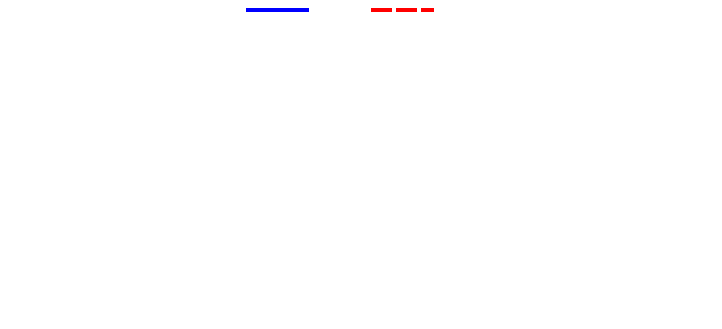}}%
    \put(0.4363959,0.44228424){\color[rgb]{0,0,0}\makebox(0,0)[lt]{\lineheight{1.5}\smash{\begin{tabular}[t]{l}$E_1$\end{tabular}}}}%
    \put(0.60959242,0.44228424){\color[rgb]{0,0,0}\makebox(0,0)[lt]{\lineheight{1.5}\smash{\begin{tabular}[t]{l}$E_2$\end{tabular}}}}%
    \put(0,0){\includegraphics[width=\unitlength,page=2]{simplif.pdf}}%
  \end{picture}%
\endgroup%

    \caption{A simplification.  Left: A vertex of an embedding~$\Gamma$ of~$G$, on the boundary of~$\complex$, incident with eight intervals.  The exterior intervals are drawn in thick lines.  Right: The result of the simplification that exchanges the intervals marked with an arrow.  The number of intervals strictly decreases.}
  \label{F:simplif}
  \end{figure}

The proof of Proposition~\ref{P:sparse} starts with any proper cellular embedding of~$\Gamma$, and iteratively changes the ordering of edges around vertices in a specific way.  Let $(E_1,E_2)$ be a partition of $E(G)$, let $v$ be a vertex of~$G$, and let $C$ be a link component at~$v$ (if the image of~$v$ under~$\Gamma$ is a singular vertex, there may be several such link components).  We restrict our attention to the edges of~$\Gamma$ incident to~$v$ and belonging to~$C$, ordered cyclically or linearly as follows.  If, in~$C$, vertex~$v$ is not incident to a boundary of~$\complex$, we consider these edges in cyclic ordering around~$v$.  Otherwise, vertex~$v$ is incident to a single boundary of~$\complex$, and we break this cyclic ordering at the boundary to obtain a linear ordering.

For $i=1,2$, an \emphdef{interval} (at~$v$, relatively to~$(E_1, E_2)$) is a maximal contiguous subsequence of edges in this (cyclic or linear) ordering that all belong to~$E_i$; the interval is labelled~$i$.  This interval is \emphdef{exterior} if the ordering is linear (because $v$ is incident to a boundary of~$\complex$) and it is the first one or the last one in this linear order; otherwise, the interval is \emphdef{interior}.  \emphdef{Simplifying}~$v$ (with respect to~$(E_1,E_2)$) means changing the (cyclic or linear) ordering of the edges of~$\Gamma$ incident to~$v$ in~$C$ by one of the two following operations (Figure \ref{F:simplif}):
\begin{enumerate}
\item either exchanging two consecutive interior intervals in that ordering, in such a way that the ordering of the edges in each interval is preserved; this operation is allowed only if $v$ is incident to at least four intervals;
\item or performing the previous operation twice, on two disjoint pairs of consecutive interior intervals in that ordering; this is allowed only if $v$ is incident to at least six intervals.
\end{enumerate}

We will rely on the following lemma.

\begin{lemma}\label{L:swap-order}
  Let $\Gamma$ be a proper cellular embedding of~$G$ on~$\complex$, and let $(E_1,E_2)$ be a partition of~$E(G)$.  Let $\Gamma'$ be another proper cellular embedding of~$G$, obtained from~$\Gamma$ by simplifying one or two vertices with respect to~$(E_1,E_2)$, while keeping the orderings of the other vertices unchanged.  Then:
  \begin{enumerate}
      \item $|E(\Pi(\Gamma',E_1,E_2))|<|E(\Pi(\Gamma,E_1,E_2))|$;
      \item for each partition $(\tilde E_1,\tilde E_2)$ of~$E(G)$ such that $\tilde E_i\subseteq E_j$ for some $i,j\in\{1,2\}$, we have $|E(\Pi(\Gamma',\tilde E_1,\tilde E_2))|\le|E(\Pi(\Gamma,\tilde E_1,\tilde E_2))|$.
  \end{enumerate}
\end{lemma}
\begin{proof}
  We first observe that the number of half-edges of~$\Pi(\Gamma,E_1,E_2)$ at~$v$ in the link component~$C$ equals twice the number of intervals associated to $(E_1,E_2)$ at~$v$ in~$C$.  (This fact will be reused later.)

  Also, we observe that a simplification at~$v$ decreases the number of intervals at~$v$ (actually, by at least two), because the (cyclic or linear) number of alternations between $E_1$ and~$E_2$ decreases.  This implies the first point of the lemma.

  For the second point, let us consider, in the (cyclic or linear) ordering around~$v$ in~$C$, a maximal contiguous sequence of edges in~$\tilde E_i$.  Since $\tilde E_i\subseteq E_j$, when simplifying with respect to~$(E_1,E_2)$, this sequence is still contiguous in the new embedding~$\Gamma'$.  If the ordering is cyclic (no boundary incident to~$v$ at~$C$), it follows that the number of intervals associated to~$(\tilde E_1,\tilde E_2)$ does not increase when replacing~$\Gamma$ with~$\Gamma'$.  If the ordering is linear, this fact is still true, because the first and last edges in this ordering are unaffected (since only interior intervals are exchanged).  This concludes the proof.
\end{proof}

\subsection{Rearranging~$\Gamma$ with respect to an edge partition}

We can now prove the following lemma:
\begin{lemma}\label{L:ordered-part}
  Let $\Gamma$ be a proper cellular embedding of~$G$ on~$\complex$, and let $(E_1,E_2)$ be a partition of~$E(G)$.  There exists a proper cellular embedding $\Gamma'$ of~$G$ such that:
  \begin{itemize}
      \item $|E(\Pi(\Gamma',E_1,E_2))|\le75c+27w$, where $w$ is the size of $\mmid(E_1,E_2)$;
      \item for each partition $(\tilde E_1,\tilde E_2)$ of~$E(G)$ such that $\tilde E_i\subseteq E_j$ for some $i,j\in\{1,2\}$, we have $|E(\Pi(\Gamma',\tilde E_1,\tilde E_2))|\le|E(\Pi(\Gamma,\tilde E_1,\tilde E_2))|$.
  \end{itemize}
\end{lemma}
\begin{proof}
  Here is an overview of the proof.  Let $\Pi:=\Pi(\Gamma,E_1,E_2)$.  We will assume that $\Pi$ has ``many monogons or bigons'' (in a sense made precise below) and show that there is another cellular embedding~$\Gamma'$ of~$G$ such that:
  \begin{itemize}
      \item $|E(\Pi(\Gamma',E_1,E_2))|<|E(\Pi(\Gamma,E_1,E_2))|$;
      \item for each partition $(\tilde E_1,\tilde E_2)$ of~$E(G)$ such that $\tilde E_i\subseteq E_j$ for some $i,j\in\{1,2\}$, we have $|E(\Pi(\Gamma',\tilde E_1,\tilde E_2))|\le|E(\Pi(\Gamma,\tilde E_1,\tilde E_2))|$.
  \end{itemize}
  By repeatedly iterating this argument, and up to replacing $\Gamma$ with~$\Gamma'$, this implies that we can assume without loss of generality that $\Pi$ has ``not too many monogons or bigons''.  We will then show that this latter property implies that $\Pi$ has at most $75c+27w$ edges, which concludes.

  \begin{figure}
  \centering
  \def\svgwidth{.5\linewidth}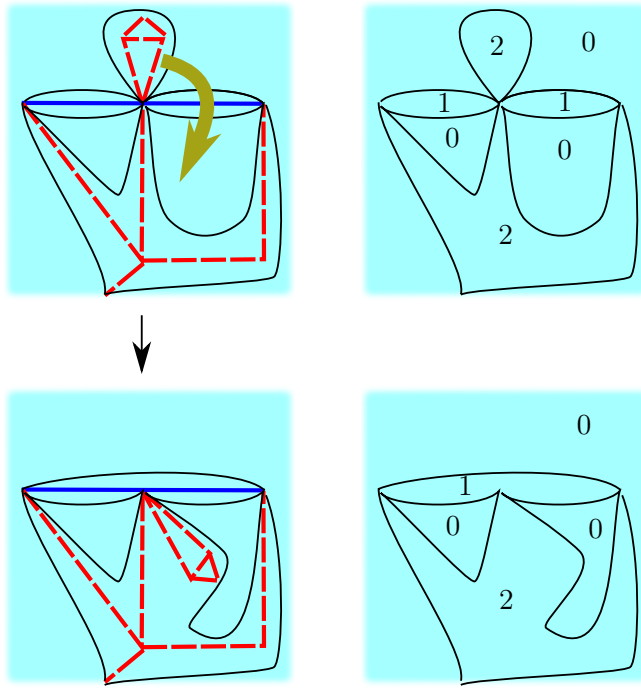
  \caption{Decreasing the number of monogons in~$\Pi$.}
  \label{F:monogon}
  \end{figure}
  First, let $v$ be a vertex of~$\Pi$, and let $C$ be a link component of~$\complex$ at~$v$ in~$\Pi$.  By construction of~$\Pi$, no monogon of~$\Pi$ is labelled~0.  Let us assume that $v$ has at least 8~incident half-edges of~$\Pi$ in~$C$ (and thus at least four intervals), and that $\Pi$ has, at~$v$ in~$C$, a monogon enclosing an interior interval.  In particular, a non-empty subgraph of~$\Gamma$ lies inside the monogon, attached to the rest of~$\Gamma$ only by~$v$, and corresponds to an interior interval~$s$ of~$\Gamma$ at~$v$ in~$C$.  We move the part of~$\Gamma$ that lies inside the monogon on the other side of the edges comprising an adjacent interior interval~$s'$ (which exists because $\Pi$ has at least four intervals at~$v$ in~$C$); see Figure~\ref{F:monogon}.  This simplifies~$v$ by swapping~$s$ with~$s'$, because $\Gamma$ has at least four intervals at~$v$ in~$C$.  Note that there is no singular vertex in the interior of the monogon, because there would be a vertex of~$\Pi$ located on the singular vertex.  The resulting graph embedding~$\Gamma'$ is still proper and cellular, and satisfies the desired properties, by Lemma~\ref{L:swap-order}.
  \begin{figure}
  \centering
  \def\svgwidth{.8\linewidth}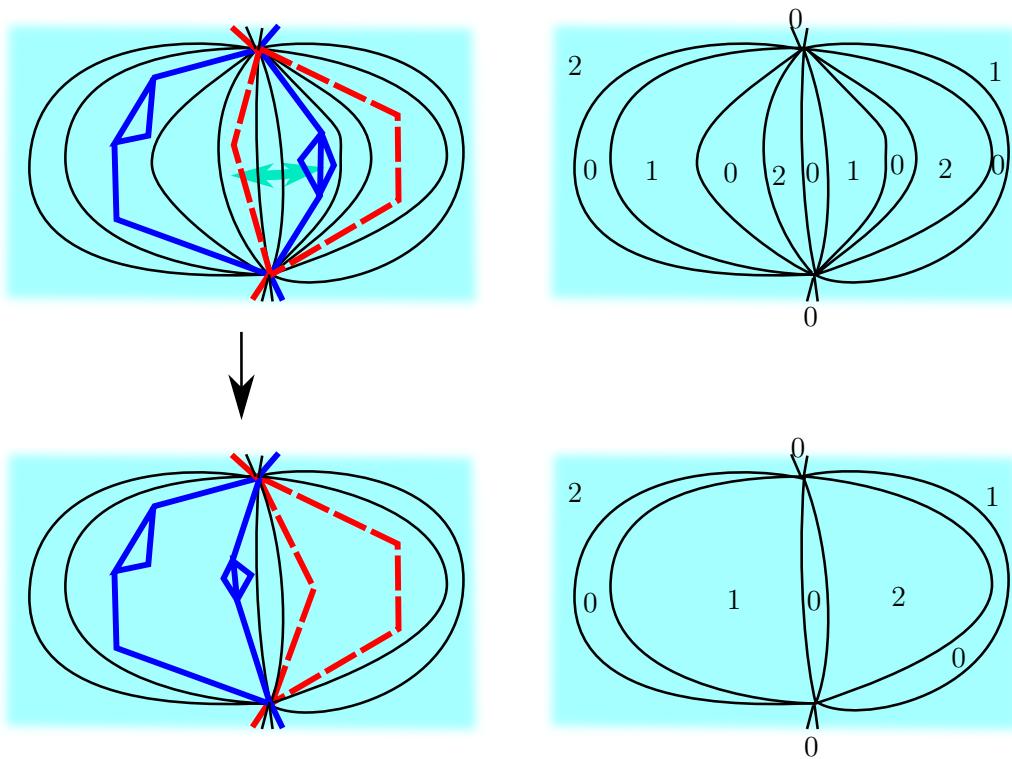
  \caption{Decreasing the number of bigons in~$\Pi$.}
  \label{F:bigons}
  \end{figure}

  Now, let us assume that $\Pi$ contains a sequence of bigons $B_1,\ldots,B_8$ such that $B_i$ and~$B_{i+1}$ share an edge for each~$i$.  So without loss of generality, we can assume that $B_1$ and~$B_5$ are labelled~$1$, $B_3$ and~$B_7$ are labelled~$2$, and the other bigons are labelled~$0$.  We modify~$\Gamma$ by exchanging the parts of~$\Gamma$ inside $B_3$ and~$B_5$; see Figure~\ref{F:bigons}.  The resulting embedding~$\Gamma'$ is also proper and cellular.  This operation simplifies the endpoints $u$ and~$v$ of these bigons.  (If $u=v$, this is Case~2 from the definition of a simplification.)
  
  We can iterate the above procedures, but only finitely many times because $|E(\Pi(\Gamma,E_1,E_2))|$ strictly decreases at each step (Lemma~\ref{L:swap-order}).  We have proved that, up to changing our initial embedding~$\Gamma$, we can assume without loss of generality the following: (i) Let $v$ be a vertex of~$\Pi$ that is incident, in a link component~$C$, to a monogon of~$\Pi$ corresponding to an interior interval; then $v$ is incident to at most 7 half-edges in~$C$; (ii) $\Pi$ has no sequence of 8~consecutive bigons as above.  To conclude, it suffices to prove that any graph~$\Pi$ satisfying these conditions has at most $75c+27w$~edges.
  
  We modify~$\Pi$ by iteratively removing all edges that are monogons, and then by iteratively replacing each bigon with a single edge, when the edges bounding the bigon are distinct.  The removal of monogons does not create any sequence of 8~consecutive bigons, because monogons are attached to vertices of degree at most~7 in their link component or correspond to exterior intervals.  So in the first step, for each vertex~$v$ and each link component~$C$ of~$v$, at most three monogons in~$C$ attached to~$v$ are removed; the number of such monogons is at most $3w$ (for the vertices of~$\Pi$ not on a singular vertex of~$\complex$) plus $3c$ (for the vertices of~$\Pi$ on singular vertices of~$\complex$).  In the second step, the number of edges is divided by at most~8.  Thus, if $\Pi'$ denotes the new embedding, we have:
  \begin{equation}
    |E(\Pi)|\le 3(c+w)+ 8|E(\Pi')|.\label{E:evpi}
  \end{equation}
  
  We now bound the number of edges of~$\Pi'$.  For this purpose, let $\surf$ be the detached surface of~$\complex$, and let $\Pi''$ be the graph naturally corresponding to~$\Pi'$ on~$\surf$ (see Section~\ref{S:datastruct-graphs}).  Any bigon of~$\Pi''$ whose boundary consists of the same edge repeated twice is itself a connected component of~$\surf$: either a sphere, in which case the corresponding connected component of~$\Pi''$ is made of two vertices and a single edge, or a projective plane, in which case the corresponding connected component of~$\Pi''$ is made of a single vertex and a single edge.  Thus, in these connected components, the number of edges of~$\Pi''$ is at most the number of vertices.  Let $\surf_0$ be obtained from~$\surf$ by removing these connected components, and $\Pi''_0$ the restriction of~$\Pi''$ to~$\surf_0$.

  $\Pi''_0$ has no monogon or bigon.  Let $\bar\surf_0$ be obtained from~$\surf_0$ by attaching a handle to each boundary component; it has a natural cellular graph embedding with at most $2c$~edges, and thus genus at most~$2c$ by Euler's formula; the graph~$\Pi''_0$ corresponds to an embedding of a graph~$\bar\Pi''_0$ on~$\bar\surf_0$, with no monogon or bigon.  Lemma~\ref{L:nogon} applied to the restriction of~$\bar\Pi''_0$ to each connected component of~$\bar\surf_0$ implies that $|E(\bar\Pi''_0)|\le 6c+3|V(\bar\Pi''_0)|$.

  Thus, $|E(\Pi')|=|E(\Pi'')|\le6c+3|V(\Pi'')|$.  Moreover, $|V(\Pi'')|\le c+w$.  Now, Inequality~(\ref{E:evpi}) implies that $|E(\Pi)|\le3(c+w)+8(6c+3(c+w))=75c+27w$, as desired.
\end{proof}

\subsection{Proof of Proposition~\ref{P:sparse}}

\begin{proof}[Proof of Proposition~\ref{P:sparse}]
  Let $\BB$ be a rooted branch decomposition of~$G$, and let $\Gamma$ be a proper cellular embedding of~$G$ on~$\complex$.  We consider each arc~$\alpha$ of the rooted branch decomposition in turn, in an arbitrary order.  For each such arc, we modify~$\Gamma$ by applying Lemma~\ref{L:ordered-part}.  We first claim that after these iterations, for each arc~$\alpha$ of~$\BB$, we have $|E(\Pi(\Gamma,\alpha))|\le75c+27w$.
  
  First, immediately after applying the above procedure to an arc~$\tilde\alpha$ of~$\BB$, corresponding to the partition~$(\tilde E_1,\tilde E_2)$ of~$E(G)$, we have $|E(\Pi(\Gamma,\tilde E_1,\tilde E_2))|\le75c+27w$.  We now prove that subsequent applications of Lemma~\ref{L:ordered-part} to other arcs of the rooted branch decomposition do not increase this number of edges.  Indeed, let $\alpha$ be another arc, corresponding to the partition~$(E_1,E_2)$ of~$E(G)$, to which we apply Lemma~\ref{L:ordered-part}.  The arc~$\alpha$ partitions the nodes of the tree~$\BB$ into two sets $N_1$ and~$N_2$, and similarly $\tilde\alpha$ partitions the nodes of the tree~$\BB$ into two sets $\tilde N_1$ and~$\tilde N_2$.  Because $\BB$ is a tree, we have $\tilde N_i\subseteq N_j$ for some $i,j\in\{1,2\}$.  This implies that $\tilde E_i\subseteq E_j$ for some $i,j\in\{1,2\}$; so the second item of Lemma~\ref{L:ordered-part} implies that the number of edges of~$\Pi(\Gamma,\tilde E_1,\tilde E_2)$ does not increase when processing arc~$\alpha$.  This proves the claim.
  
  Finally, there remains to prove that, for each internal node~$\nu$ of~$\BB$, the graph $\Pi(\Gamma,\nu)$ has at most $3(75c+27w)/2$ edges.  Let $(E_1,E_2,E_3)$ be the edge partition associated with~$\nu$.  By the claim we have just proved, it suffices to prove that the number of edges of this graph is at most half the sum of the numbers of edges of $\Pi(\Gamma,E_1,E_2\cup E_3)$, $\Pi(\Gamma,E_1\cup E_2,E_3)$, and $\Pi(\Gamma,E_1\cup E_3,E_2)$.
  
  Let $v$ and~$C$ be as above.  Any half-edge of~$\Pi(\Gamma,\nu)$ incident with~$v$ and on~$C$ arises, in cyclic order around~$v$, (1) either between two half-edges of~$\Gamma$, or (2) between a half-edge of~$\Gamma$ and a boundary part of~$\complex$ incident with~$v$ on~$C$.  We consider each of these cases in turn below.

  Between two half-edges of~$\Gamma$, there are either zero or two half-edges in~$\Pi(\Gamma,\nu)$.  In the latter case, this means that these two consecutive half-edges of~$\Gamma$ are in two different sets $E_i$ and~$E_j$.  Thus, between these two consecutive half-edges of~$\Gamma$, necessarily two half-edges appear in exactly two of $\Pi(\Gamma,E_1,E_2\cup E_3)$, $\Pi(\Gamma,E_1\cup E_2,E_3)$, and $\Pi(\Gamma,E_1\cup E_3,E_2)$.

  Between a half-edge and a boundary part, there is a single half-edge in~$\Pi(\Gamma,\nu)$, and also in~$\Pi(E_1,E_2)$, in $\Pi(E_1,E_3)$, and in $\Pi(E_2,E_3)$.

  This concludes the proof of the proposition.
\end{proof}

\section{Dynamic programming algorithm}\label{S:algo}

The result of this section is the following proposition.
\begin{proposition}\label{P:algo}
  Let $\complex$ be a 2-complex and $G$ a graph, which satisfy the properties of Proposition~\ref{P:preproc}.  Let $c$ be the size of~$\complex$ and $n$ the number of vertices and edges of~$G$.  Let $\BB$ be a rooted branch decomposition of~$G$ of width~$w$.
  In $(c+w)^{O(c+w)}n$ time, one can report one of the following statements, which is true:
  \begin{itemize}
      \item $G$ has no sparse proper cellular embedding into~$\complex$;
      \item $G$ has an embedding into~$\complex$.
  \end{itemize}
\end{proposition}

(Proposition~\ref{P:sparse} implies that we can remove the adjective ``sparse'' in the above proposition.)

\subsection{Bounding graphs}

Let $\BB$ be a rooted branch decomposition of~$G$ of width~$w$.   Recall (see Section~\ref{S:branchdecomp}) that the root of~$\BB$ is a leaf associated to no edge of~$G$.  Our algorithm will use dynamic programming in the rooted branch decomposition.  In the following, we fix an arc~$\alpha$ of~$\BB$, and we let $G_\alpha$ be the subgraph of~$G$ induced by the edges of~$G$ corresponding to the leaves of the subtree of~$\BB$ rooted at~$\alpha$.  The general idea is that we compute all possible relevant embeddings of~$G_\alpha$ in subregions of~$\complex$.  Such subregions will be delimited by a graph embedded on~$\complex$ of small complexity.  For the dynamic program to work, we also need to keep track of the location of the vertices in the middle set of~$\alpha$.  More precisely, we define \emph{bounding graphs} below; partitioning graphs are, in particular, bounding graphs.  A \emphdef{bounding graph} for~$G_\alpha$ is a proper labelled graph embedding~$\Pi$ on~$\complex$ (but possibly non-cellular), such that:
\begin{itemize}
    \item some vertices of~$\Pi$ are labelled; these labels are exactly $\mmid(\alpha$), and each label appears exactly once;
    \item each unlabelled vertex of~$\Pi$ is mapped to a singular vertex of~$\complex$;
    \item each face of~$\Pi$ is labelled 0, 1, or~2;
    \item $G_\alpha$ has an embedding~$\Gamma_\alpha$ that \emphdef{respects}~$\Pi$:  each vertex of~$\Pi$ labelled~$v$ is mapped, under~$\Pi$, to the image of~$v$ in~$\Gamma_\alpha$; moreover, the relative interior of each edge of~$\Gamma_\alpha$ lies in the interior of a face of~$\Pi$ labelled~2.
\end{itemize}
(It may seem strange to have three labels for the faces of a bounding graph, but it simplifies the argumentation, although two would suffice.)

A bounding graph for~$G_\alpha$ is \emphdef{sparse} if it has at most $75c+27w$~edges.  Remark that, if $\Gamma$ is a sparse proper cellular embedding of~$G$ on~$\complex$ (as defined in Section~\ref{S:partitioning}), then $\Pi(\Gamma,\alpha)$ is a sparse bounding graph for the restriction of~$\Gamma$ to~$G_\alpha$.

Henceforth, we regard two (labelled) proper graph embeddings as equal if and only if their (labelled) combinatorial maps are isomorphic.  This convenient abuse of language is legitimate because whenever two proper graph embeddings have the same combinatorial map, there is an ambient self-homeomorphism of~$\complex$ that maps one into the other.

A list $\LL_\alpha$ of sparse bounding graphs for~$G_\alpha$ is \emphdef{exhaustive} if the following condition holds:  If $G$ has a sparse proper cellular embedding on~$\complex$, then for each such embedding~$\Gamma$, the (combinatorial map of the) graph~$\Pi(\Gamma,\alpha)$ is in~$\LL_\alpha$.

The induction step for the dynamic programming algorithm is the following.
\begin{proposition}\label{P:dynprog}
  There is a universal constant~$K>0$ such that the following holds.  Let $\nu$ be a non-root node of~$\BB$ and $\alpha$ be the arc of~$\BB$ incident to~$\nu$ that is the closest to the root~$\rho$.  Assume that, for each arc~$\beta\ne\alpha$ of~$\BB$ incident to~$\nu$, we are given an exhaustive list of at most $(c+w)^{K(c+w)}$ sparse bounding graphs for~$G_\beta$.  Then we can, in $(c+w)^{O(c+w)}$ time, compute an exhaustive list of at most $(c+w)^{K(c+w)}$ sparse bounding graphs for~$G_\alpha$.
\end{proposition}

Assuming Proposition~\ref{P:dynprog}, the proof of which is deferred to the next subsection, it is easy to prove Proposition~\ref{P:algo}:
\begin{proof}[Proof of Proposition~\ref{P:algo}, assuming Proposition~\ref{P:dynprog}]
  We apply the algorithm of Proposition~\ref{P:dynprog} in a bottom-up manner in the rooted branch decomposition~$\BB$.  Note that, as stated, Proposition~\ref{P:dynprog} handles the base case of the recursion, namely, the case where $\nu$ is a leaf.  So let $\alpha$ be the arc of~$\BB$ incident with the root node~$\rho$.  We end up with an exhaustive list of at most $(c+w)^{K(c+w)}$ sparse bounding graphs for~$G_\alpha=G$.  By definition of a bounding graph, if this list is non-empty, then $G$ has an embedding on~$\complex$.  On the other hand, by definition of an exhaustive list, if this list is empty, then $G$ has no sparse proper cellular embedding on~$\complex$.
  
  There are $O(n)$ recursive calls, each of which takes $(c+w)^{O(c+w)}$ time.
\end{proof}

\subsection{The induction step: Proof of Proposition~\ref{P:dynprog}}

\begin{proof}[Proof of Proposition~\ref{P:dynprog}]
  By Lemma~\ref{L:enum-embeddings}, let $K>0$ be such that the number of combinatorial maps of proper embeddings of graphs with at most $75c+27w$ vertices and at most at most $75c+27w$ edges on~$\complex$ is at most $(c+w)^{K(c+w)}$.

  \emph{First case.}
  Let us first assume that $\nu$ is a (non-root) leaf of~$\BB$; thus, $G_\alpha$ is a single edge~$uv$.  We will compute the labelled combinatorial maps of \emph{all} sparse bounding graphs for~$G_\alpha$.  It is clear that this will be an exhaustive list.  Indeed, assume that $G$ has a sparse proper cellular embedding~$\Gamma$ on~$\complex$; by sparsity, $\Pi(\Gamma,\alpha)$ has at most $75c+27w$~edges; thus, one of the labelled combinatorial maps computed will be equal to that of~$\Pi(\Gamma,\alpha)$.

  So let us describe how to enumerate all the labelled combinatorial maps of bounding graphs for~$G_\alpha$.  Using Lemma~\ref{L:enum-embeddings}, we enumerate all possible labelled (combinatorial maps of) proper graph embeddings~$\Pi$ on~$\complex$ such that:
  \begin{itemize}
      \item two vertices of~$\Pi$ are labelled $u$ and~$v$ (or a single vertex of~$\Pi$, if $u=v$); the other vertices are unlabelled; the singular vertices of~$\complex$ are covered by the vertices of~$\Pi$; conversely, every vertex of~$\Pi$, except perhaps $u$ and/or~$v$, is mapped to a singular vertex of~$\complex$;
      \item $\Pi$ has at most $75c+27w$ edges;
      \item each face of~$\Pi$ is labelled 0, 1, or~2;
      \item $\Pi$ has a face labelled~2 whose boundary contains both vertices $u$ and~$v$.
  \end{itemize}
  It is clear that these labelled combinatorial maps represent all the sparse bounding graphs for~$G_\alpha$.

  \medskip

  \emph{Second case.}
  Let us now assume that $\nu$ is an internal node of~$\BB$.  As above, let $\alpha$ be the arc of~$\BB$ incident to~$\nu$ that is the closest to the root~$\rho$.  Let $\beta$ and~$\gamma$ be the arcs different from~$\alpha$ incident to~$\nu$.  Let $\LL_\beta$ and~$\LL_\gamma$ be exhaustive lists of bounding graphs for~$G_\beta$ and~$G_\gamma$, respectively.  Intuitively, every pair of bounding graphs in $\LL_\beta$ and $\LL_\gamma$ that are compatible, in the sense that the regions labelled~2 in each of these two graphs are disjoint, will lead to a bounding graph in~$\LL_\alpha$.  This is the motivating idea to our approach.  More precisely, we will enumerate labelled combinatorial maps~$\Pi$, each of which can be ``restricted'' to two compatible graphs, which are possible bounding graphs for $G_\beta$ and~$G_\gamma$.  If these two restrictions lie in $\LL_\beta$ and~$\LL_\gamma$, this leads to a bounding graph that is added to~$\LL_\alpha$.  Even more precisely, we observe that, given a sparse, proper, cellular embedding~$\Gamma$ of~$G$ and a node~$\nu$ of~$\BB$, from $\Pi(\Gamma,\nu)$ we can compute $\Pi(\Gamma,\delta)$ for each incident arc~$\delta$.  (See Figures \ref{F:partitioning} and~\ref{F:partitioning-new-node}; we describe this in detail below.)  So, from $\Pi$, we compute its ``restrictions'' to compatible potential bounding graphs for $G_\beta$ and~$G_\gamma$, with the property that if $\Pi$ is actually of the form $\Pi(\Gamma,\nu)$, we indeed compute $\Pi(\Gamma,\beta)$ and~$\Pi(\Gamma,\gamma)$.
  
  We first introduce some terminology.  Let $\Pi$ be a graph properly embedded on~$\complex$ (possibly non-cellularly), with faces labelled 0, 1, 2, or~3, and with labels on some vertices.  We denote by~$\Pi^-$ the map obtained from~$\Pi$ by replacing each label~3 on a face by a~2.  Let $i,j,k$ be integers such that $\{i,j,k\}=\{1,2,3\}$.  We will define a graph embedding $\Pi_{i,j}$ obtained from~$\Pi$ by somehow replacing $j$-labels with $i$-labels and  ``merging'' the resulting faces.  First, for an illustration, refer back to Figures~\ref{F:partitioning} and~\ref{F:partitioning-new-node}: If $\Pi$ is the graph embedding depicted on the right of Figure~\ref{F:partitioning-new-node}, then the configurations shown on the right of Figure~\ref{F:partitioning} correspond, from top to bottom, to $\Pi_{2,3}$, $\Pi_{1,3}$, and $(\Pi_{1,2})^-$.

  Formally, $\Pi_{i,j}$ is defined as follows. First, let us replace all face labels~$j$ by~$i$.  Now, for each face~$f$ of~$\Pi$ that is homeomorphic to an open disk, possibly with additional points on the boundary of~$\complex$, and labelled~0, we do the following.  The boundary of~$f$ is made of edges of~$\Pi$, possibly with additional points on the boundary of~$\complex$; for the sake of the discussion, let us temporarily label each such edge by the label of the face on the other side of~$f$.  If the boundary of~$f$ is entirely made of edges labelled~$i$ (without boundary points of~$\complex$), then we remove all these edges, and $f$ becomes part of a larger face labelled~$i$.  Otherwise, for each maximal subsequence $e_1,\ldots,e_\ell$ of edges along the boundary of~$f$ that are all labelled~$i$, we remove each of $e_1,\ldots,e_\ell$, and replace them with an edge inside~$f$ from the source of~$e_1$ to the target of~$e_\ell$.  Finally, we remove all isolated vertices that do not coincide with singular vertices of~$\complex$, and all vertices in the relative interior of an isolated edge that are incident to two faces with the same label.

  The easy but key properties of this construction are the following:
  \begin{itemize}
      \item[(i)] Assume that $\Pi_{1,3}$ is a bounding graph for~$G_\beta$ and $(\Pi_{1,2})^-$ is a bounding graph for~$G_\gamma$.  Then $\Pi_{2,3}$ is a bounding graph for~$G_\alpha$.
      \item[(ii)] The node~$\nu$ naturally partitions the edge set of~$G$ into three parts, which we denote by $E_1$ (on the side of~$\alpha$), $E_2$ (on the side of~$\beta$), and~$E_3$ (on the side of~$\gamma$).  Assume that $G$ has a sparse proper cellular embedding~$\Gamma$ on~$\complex$ and that $\Pi$ is the partitioning graph $\Pi(\Gamma,E_1,E_2,E_3)$.  Then:
      \begin{itemize}
      \item $\Pi(\Gamma,\alpha)=\Pi(\Gamma,E_1,E_2\cup E_3)=\Pi_{2,3}$;
      \item $\Pi(\Gamma,\beta)=\Pi(\Gamma,E_1\cup E_3,E_2)=\Pi_{1,3}$;
      \item $\Pi(\Gamma,\gamma)=\Pi(\Gamma,E_1\cup E_2,E_3)=(\Pi_{1,2})^-$.
      \end{itemize}
  \end{itemize}
  Property~(i) follows from the above definitions.  Property~(ii) follows from the definitions and from the construction of the partitioning graphs; it is, again, illustrated by Figures~\ref{F:partitioning} and~\ref{F:partitioning-new-node}: if $(E_1,E_2,E_3)$ is the edge partition depicted on Figure~\ref{F:partitioning-new-node}, then the edge partitions depicted on Figure~\ref{F:partitioning}, left, are, respectively, $(E_1,E_2\cup E_3)$, $(E_1\cup E_3,E_2)$, and $(E_1\cup E_2,E_3)$.  As shown above, the corresponding partitioning graphs are respectively $\Pi_{2,3}$, $\Pi_{1,3}$, and~$\Pi_{1,2}^-$.
  
  We compute our exhaustive list~$\LL_\alpha$ of sparse bounding graphs for~$G_\alpha$ as follows.  Initially, let this list be empty.  Using Lemma~\ref{L:enum-embeddings}, we enumerate all combinatorial maps~$\Pi$ of graphs with at most $c+3w$ vertices and $3(75c+27w)/2$ edges properly embedded on~$\complex$ (possibly non-cellularly), with faces labelled 0, 1, 2, or~3, and such that the labels appearing on the vertices are exactly the vertices of the middle set of~$\alpha$, $\beta$, or~$\gamma$ (and each label appears exactly once).  This takes $(c+w)^{O(c+w)}$ time.  Whenever $\Pi_{1,3}\in\LL_\beta$ and $(\Pi_{1,2})^-\in\LL_\gamma$, we add $\Pi_{2,3}$ to~$\LL_\alpha$.
  
  Finally, we need a postprocessing step to control the size of $\LL_\alpha$.  First, we remove the graphs that are not sparse or contain vertices that bear a label not in~$\mmid(\alpha)$.  Second, we eliminate duplicates by testing pairwise isomorphism between the labelled combinatorial maps in~$\LL_\alpha$.  This ensures that $\LL_\alpha$ has size at most $(c+w)^{K(c+w)}$.
  
  Now, $\LL_\alpha$ contains only sparse bounding graphs for~$G_\alpha$, by~(i) above.  Moreover, let $\Gamma$ be a sparse proper cellular graph embedding of~$G$ on~$\complex$.  One of the graphs~$\Pi$ enumerated in the previous paragraph is $\Pi(\Gamma,\nu)$, by sparsity and because $\Pi(\Gamma,\nu)$ has at most $c+3w$ vertices.  By definition of $\LL_\beta$ and~$\LL_\gamma$, we have that $\Pi(\Gamma,\beta)\in\LL_\beta$ and $\Pi(\Gamma,\gamma)\in\LL_\gamma$, so by~(ii) above, $\Pi(\Gamma,\alpha)\in\LL_\alpha$, which implies that $\LL_\alpha$ is exhaustive.
\end{proof}

\section{Reduction to proper cellular embeddings}\label{S:cell}

This section is devoted to proving the following result:
\begin{proposition}\label{P:cell}
  Let $\complex$ be a 2-complex with at most $c$ simplices, and $G$ a graph with at most $n$ vertices and edges and branchwidth at most~$w$.  Assume that $G$ and~$\complex$ satisfy the properties of Proposition~\ref{P:preproc}.  In $c^{O(c)}+O(cn)$ time, one can compute a graph~$G'$, and $c^{O(c)}$ 2-complexes $\complex_i$, such that:
  \begin{enumerate}
      \item each $\complex_i$ and $G'$ satisfy the properties of Proposition~\ref{P:preproc};
      \item $G'$ has at most $5cn$ vertices and~$5cn$ edges, and branchwidth at most~$w$;
      \item each $\complex_i$ has size at most~$c$;
      \item if, for some~$i$, $G'$ embeds into~$\complex_i$, then $G$ embeds into~$\complex$;
      \item if $G$ embeds into~$\complex$, then for some~$i$, $G'$ has a proper cellular embedding into~$\complex_i$.
  \end{enumerate}
\end{proposition}

There are two main ideas in the proof: (1) To make an embedding into~$\complex$ proper, it suffices to subdivide its edges and to move the embedding slightly; (2) to make a proper embedding into~$\complex$ cellular, it suffices to simplify~$\complex$, by detaching some link components, removing some isolated edges, and simplifying the topology of the surface parts; there are $c^{O(c)}$ ways to achieve this, so we can try all possibilities.

We start with auxiliary results.  Let $\surf$ be a surface (possibly disconnected, possibly with boundary).  A \emphdef{cutting operation} on~$\surf$ consists of cutting it along a simple closed curve, and attaching a disk to the resulting boundary component(s).  A cutting operation is \emphdef{essential} if the simple closed curve is non-contractible.

The following result is not hard and essentially folklore (a related but slightly weaker result is provided by Matou\v{s}ek, Sedgwick, Tancer, and Wagner~\cite[Lemma~3.1]{mstw-utsnc-16}), but we could not find a precise reference.
\begin{lemma}\label{L:cutsurf-breadth}
  Let $\surf$ be a (connected) surface with genus~$g$.  The number of possibly disconnected surfaces, up to homeomorphism, that can be obtained from~$\surf$ by a cutting operation is at most $g+3$, and we can compute them in linear time.  Moreover, this cutting operation leads either to a single surface with genus strictly smaller than~$g$, or to two surfaces, the sum of the genera of which equals~$g$, and the size of the surface (sum of the number of connected components, total genus, and number of boundary components) increases by at most one.
\end{lemma}
\begin{proof}
  This basically follows from the classification of surfaces together with Euler's formula.  A cutting operation of~$\surf$ along a closed curve~$\gamma$ falls into exactly one of the following three categories:
  \begin{enumerate}
  \item \emph{Case where $\gamma$ is separating.}  The cutting operation on~$\surf$ results in two surfaces $\surf_1$ and~$\surf_2$, in which their respective genera $g_1$ and~$g_2$ satisfy $g=g_1+g_2$.  Moreover, if $\surf$ is non-orientable, then at least one of $\surf_1$ or~$\surf_2$ is non-orientable.  Finally, all pairs of surfaces $(\surf_1,\surf_2)$ satisfying these constraints can be obtained as the result of a cutting operation on~$\surf$.
  
  \item \emph{Case where $\gamma$ is non-separating but two-sided.}  This is only possible if $g\ge2$.  The cutting operation on~$\surf$ results in a single surface~$\surf'$ with genus~$g-2$.  If $\surf$ is orientable, then so is~$\surf'$; otherwise, $\surf'$ is either orientable or non-orientable (unless of course $g=2$, in which case it is necessarily orientable, or $g$ is odd, in which case it is necessarily non-orientable).  All surfaces~$\surf'$ satisfying these constraints can be obtained.
 
  \item \emph{Case where $\gamma$ is one-sided.}  This is only possible if $\surf$ is non-orientable and $g\ge1$.  The cutting operation on~$\surf$ results in a single surface~$\surf'$ with genus~$g-1$, orientable or not (unless of course $g=1$, in which case it is orientable, or $g$ is even, in which case it is non-orientable).    All surfaces~$\surf'$ satisfying these constraints can be obtained.\qedhere
  \end{enumerate}
\end{proof}

\begin{lemma}\label{L:cutsurf}
  Let $\surf$ be a surface with $k$~connected components, total genus~$g$, and with $b$~boundary components in total.
  In $(k+g+b)^{O(g+b)}$ time, we can enumerate all $(k+g+b)^{O(g+b)}$ possibly disconnected surfaces with boundary, up to homeomorphism, arising from~$\surf$ by one or several successive \emph{essential} cutting operations.  These surfaces have $O(k+g+b)$ connected components and size $O(k+g+b)$.
\end{lemma}
\begin{proof}
  It is useful to organize the set of all surfaces (possibly disconnected, possibly with boundary) arising by essential cutting operations in a tree with root~$\surf$, in which the children of a node result from a single essential cutting operation.  We prove that (1) the depth of the tree is $O(g+b)$ and that (2) each node of the tree has $O((k+1)(g+1)(b+1))$ children, which concludes (because by Lemma~\ref{L:cutsurf-breadth}, the size of a surface increases by at most one by a cutting operation).

  Let $\surf'$ be a (possibly disconnected, possibly with boundary) surface resulting from a sequence of essential cutting operations on~$\surf$.  By Lemma~\ref{L:cutsurf-breadth}, the total genus of~$\surf'$ is at most~$g$.  Moreover, since we consider only essential cutting operations, each connected component of~$\surf'$ either has positive genus or contains at least one boundary component, unless it was itself a connected component of~$\surf$; so the number of connected components of~$\surf'$ is at most $k+g+b$.
  
  Let $\varphi(\surf')$ be equal to twice the total genus of~$\surf'$ minus its number of connected components.  By Lemma~\ref{L:cutsurf-breadth}, this potential function strictly decreases at each cutting operation.  Moreover, we have $\varphi(\surf)=2g-k$, and by the previous paragraph $\varphi(\surf')$ is at least $-(k+g+b)$.  This proves~(1).
  
  By Lemma~\ref{L:cutsurf-breadth}, for any surface of genus~$g$ without boundary, there are at most $g+3$ ways of performing a cutting operation up to homeomorphism.  After a sequence of essential cutting operations, we have a surface~$\surf'$ with at most $k+g+b$ connected components, with total genus at most~$g$, and with $b$~boundary components.  The number of surfaces that can be obtained from~$\surf'$ by a cutting operation is at most $k(g+3)(b+1)$, since we first choose which connected component to cut along, the way to cut it ignoring the boundary components, and the number of boundary components in each connected component (if the cut is separating).
\end{proof}

We can now conclude the proof of this section.
\begin{proof}[Proof of Proposition~\ref{P:cell}]
  First, if the detached surface of~$\complex$ is non-empty (equivalently, if $\complex$ has at least one triangle), then we test the planarity of each connected component of~$G$~\cite{ht-ept-74} in $O(n)$ time, and remove every connected component of~$G$ that is planar; obviously, this does not affect the embeddability of~$G$ on~$\complex$.
  
  In a second step, we split each isolated edge of~$\complex$ into five isolated edges.  Then, for each subset~$S$ of the isolated edges, we build a new 2-complex obtained from~$\complex$ by removing~$S$.   We obtain $2^{O(c)}$ 2-complexes, each of size $O(c)$.  The input graph~$G$ embeds on~$\complex$ if and only if it embeds into one of these 2-complexes; moreover, if $G$ embeds on~$\complex$, it embeds into one of the 2-complexes in such a way that every isolated edge is covered by the embedding.  (Indeed, remember that $G$ has at most one connected component that is a path.)
  
  We now iteratively dissolve every degree-two vertex of~$G$, and then subdivide $5c$ times each edge of~$G$.  This new graph~$G'$ has at most $5cn$ vertices and edges, and branchwidth at most~$w$.  Clearly, $G$ embeds on~$\complex$ if and only if $G'$ embeds in one of the 2-complexes defined in the previous paragraph; moreover, if $G$ embeds on~$\complex$, then $G'$ has an embedding on one of these 2-complexes in which the relative interior of every edge of~$G'$ contains no singular vertices of~$\complex$ (and, as above, such that every isolated edge is covered by the embedding).
  
  Each singular vertex~$v$ of each of these 2-complexes is incident to at least two link components.  For each such singular vertex~$v$ and for each partition of the link components at~$v$, we replace $v$ with new vertices, one for each element in the partition; two link components at~$v$ stay adjacent via one of these new vertices if and only if these link components are in the same part.  We obtain $c^{O(c)}$ 2-complexes, each of size $O(c)$.  The input graph~$G$ embeds on~$\complex$ if and only if $G'$ embeds in one of these 2-complexes; moreover, if $G$ embeds on~$\complex$, then $G'$ has an embedding into one of these 2-complexes in which every link component of each singular vertex~$v$ is used by an edge of~$G$ connected to~$v$ in that link component (and, as above, such that every isolated edge is covered by the embedding, and such that the relative interior of every edge is distinct from any singular vertex).
  
  Every embedding of a graph into a 2-complex can be perturbed so that it avoids the boundary of the 2-complex, except possibly at singular vertices. This means that, if $G$ embeds on~$\complex$, then $G'$ has a proper cellular embedding into one of the 2-complexes built in the previous paragraph, except that the faces of~$G'$ may fail to be disks, but are more general (connected) surfaces with boundary.
  
  To dispense ourselves from this latter exception, we need to build more 2-complexes.  This case occurs only if the detached surface is non-empty, so by our earlier preprocessing, we can assume that $G'$ contains no planar connected component, and so has $O(c)$ connected components (because, by Lemma~\ref{L:genus-oversurface}, in order for $G'$ to be embeddable on a 2-complex of size $O(c)$, it must have genus $O(c)$).
 
  The detached surface~$\surf$ is a surface (possibly disconnected, possibly with boundary); the trace of the set of singular vertices on~$\surf$ corresponds to marked vertices, some in the interior of~$\surf$, some on the boundary.  Henceforth, we regard the former ones as small boundary components.  For each 2-complex obtained above, we consider, up to homeomorphism, all 2-complexes arising from zero, one, or several essential cutting operations on~$\surf$, and then by removing an arbitrary subset of the connected components of the resulting surface.  Up to homeomorphism, by Lemma~\ref{L:cutsurf}, there are $c^{O(c)}$ ways of cutting~$\surf$; since we consider all 2-complexes obtained up to homeomorphism, we need to consider each boundary component of the detached surface~$\surf$ as labelled (which is not the case in Lemma~\ref{L:cutsurf}); however, this only adds a factor of $c^{O(c)}$.  In total, we obtain, in $c^{O(c)}$ time, $c^{O(c)}$ 2-complexes, each of size $O(c)$.  Then, for each such 2-complex, we consider all possible ways of removing an arbitrary subset of connected components of the 2-complex; the number of the resulting 2-complexes is still~$c^{O(c)}$.  By construction, the input graph~$G$ embeds on~$\complex$ if and only if $G'$ embeds in one of these 2-complexes.  Moreover, assume that it is the case; as shown above, $G'$ has a proper embedding into one of the 2-complexes in the previous paragraph, except that faces of~$G'$ are (connected) surfaces, not necessarily disks.  Whenever a face has non-empty boundary and is not homeomorphic to a disk, we perform an essential cutting operation of that face along a closed curve inside that face; the closed curve along which we cut is non-contractible in~$\surf$, because otherwise it would bound a disk in~$\surf$, which would itself contain a planar connected component of~$G'$, and we have shown above that we may assume that no such component exists.  After iterating this operation as much as possible, every face of~$G'$ in the resulting 2-complex is either is a disk or has empty boundary; in the latter case, $G'$ avoids the corresponding connected component, so we can simply remove it.  Eventually, after a number of essential cutting operations of~$\surf$ and removing some connected components of the 2-complex, the embedding of~$G'$ is cellular in one of the $c^{O(c)}$ 2-complexes of size $O(c)$ enumerated above.

  By construction, each $\complex_i$ and~$G'$ satisfy the properties of Proposition~\ref{P:preproc}.
\end{proof}

\section{Algorithm for bounded branchwidth: Proof of Theorem~\ref{T:bounded-bw}}\label{S:bounded-bw}

We can now combine the ingredients of the previous sections to describe our algorithm for bounded branchwidth.
\begin{proof}[Proof of Theorem~\ref{T:bounded-bw}]
  By Proposition~\ref{P:preproc}, we can assume that $\complex$ has no 3-book and no connected component that is reduced to a single vertex, and that $G$ has no connected component reduced to a single vertex and at most one connected component that is a path.  If necessary, we convert the combinatorial description of~$\complex$ into the topological data structure (Section~\ref{S:datastruct-2complexes}).

  We apply Proposition~\ref{P:cell}.  In $c^{O(c)}+O(cn)$ time, we obtain a graph~$G'$ and a set of $c^{O(c)}$ 2-complexes~$\complex_i$ such that:
  \begin{enumerate}
      \item each $\complex_i$ and $G'$ satisfy the properties of Proposition~\ref{P:preproc};
      \item $G'$ has at most $5cn$ vertices and~$5cn$ edges, and branchwidth at most~$w$;
      \item each $\complex_i$ has size at most~$c$;
      \item\label{E:1} if, for some~$i$, $G'$ embeds into~$\complex_i$, then $G$ embeds into~$\complex$;
      \item\label{E:2} if $G$ embeds into~$\complex$, then for some~$i$, $G'$ has a proper cellular embedding into~$\complex_i$.
  \end{enumerate}

  We then run the algorithm from Proposition~\ref{P:algo} in each of the instances $(\complex_i,G')$, in total time ${(c+w)}^{O(c+w)}n$.  This algorithm correctly reports either that $G'$ has no sparse proper cellular embedding into~$\complex_i$ or that $G'$ has an embedding into~$\complex$.  If for at least one of these instances, the algorithm reports that $G'$ embeds into~$\complex_i$, then we report that $G$ embeds into~$\complex$.  Otherwise, we report that $G$ does not embed into~$\complex$.
  
  There remains to prove that the algorithm is correct.  If our algorithm reports that $G$ embeds into~$\complex$, then it is obviously indeed the case (Property~(\ref{E:1}) above).  Conversely, let us assume that $G$ has an embedding into~$\complex$.  Thus, by Property~(\ref{E:2}) above, let $i$ be such that $G'$ has a proper cellular embedding into~$\complex_i$.  By Proposition~\ref{P:sparse}, $G'$ also has such an embedding into~$\complex_i$ that is sparse.   Thus, the algorithm in Proposition~\ref{P:algo} (correctly) reports that $G'$ has an embedding into~$\complex_i$, and finally our overall algorithm reports that $G$ has an embedding into~$\complex$.
\end{proof}

\section{Reduction to bounded branchwidth: Proof of Theorem~\ref{T:irrelevant}}\label{S:irrelevant}

This section is devoted to the proof of Theorem~\ref{T:irrelevant}.  The proof technique is based on an irrelevant vertex method; we borrow ingredients from Kociumaka and Pilipczuk~\cite[Section~5]{kp-dvgbg-19} and from graph minor algorithms, but some new arguments are needed, in particular in the beginning of the proof of Proposition~\ref{P:irrelevant_0}, and because we need to take singular vertices into account.

\subsection{Finding a large planar part}

A \emphdef{wall} of size $k\times k$ is a subgraph of the $(k\times k)$-grid obtained by removing alternately the vertical edges of even (resp.\ odd) $x$-coordinate in each even (resp.\ odd) line, and then the degree-one vertices; see Figure~\ref{F:wall}.

\begin{figure}
\centering
\includegraphics[width=.2\linewidth]{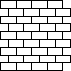}
\caption{A $(10\times 10)$-wall.}
\label{F:wall}
\end{figure}

As an intermediate goal towards the proof of Theorem~\ref{T:irrelevant}, we will prove in this subsection:
\begin{proposition}\label{P:irrelevant_0}
Let $G$ be a graph with $n$ vertices and edges and $g\ge2$ be an integer.  In $O(2^{\poly(g)}n)$ time, we can do one of the following:
\begin{enumerate}
\item compute a rooted branch decomposition of $G$ of width~$g^{O(1)}$;
\item correctly report that $G$ has genus at least~$g$;
\item or compute a cycle~$\gamma$ of $G$ such that one connected component of $G-\gamma$ is planar and contains a subdivision of the $(g\times g)$-wall, which is also computed.
\end{enumerate}
\end{proposition}
(Our proof implies a degree of $18+\varepsilon$, for any $\varepsilon>0$, for the polynomials in~$g$ in the above statement.)

We will use the following lemma, whose proof relies on some advanced results related to treewidth and tree decompositions.  It is probably folklore in some circles, and some similar ideas have been used recently~\cite[Section~3.2]{sst-mavfw-24}, but we could not find a reference.
\begin{lemma}\label{L:branch-or-grid}
  Let $G$ be a graph with $n$ vertices and edges, and $k\ge2$.  Then, one can compute either a rooted branch decomposition of~$G$ of width polynomial in~$k$ or a $(k\times k)$-grid minor of~$G$, in time $O(2^{\poly(k)}\cdot n)$.
\end{lemma}
(Our proof implies a degree of $9+\varepsilon$, for any $\varepsilon>0$, for the polynomials in~$k$ in the above statement.)
\begin{proof}
  In the same way as the branchwidth of a graph~$H$ is the minimum width of a branch decomposition of~$H$, the treewidth of~$H$ is the minimum width of a tree decomposition of~$H$.  We do not define tree decompositions here, but it suffices to know that one can, in linear time, convert a branch decomposition of~$G$ of width~$w$ into a tree decomposition of~$G$ of width at most~$w+1$; and similarly in linear time we can convert a tree decomposition of~$G$ of width~$w$ into a branch decomposition of~$G$ of width at most $3w/2-1$~\cite[Theorem~5.1]{rs-gm10o-91}.
  
  By a result by Chekuri and Chuzhoy~\cite{cc-pbgmt-16}, we know that for some large enough constants $c,d$, if $G$ has treewidth at least $ck^d$, then it has a $(k\times k)$-grid minor.  One can take $d=9+\varepsilon$, by a result by Chuzhoy and Tan~\cite{ct-ttbeg-21}.

  We will use an algorithm to approximate the treewidth by Bodlaender, Drange, Dregi, Fomin, Lokshtanov, and Mi.~Pilipczuk~\cite{bddflp-c5aat-16}: Given a graph~$G$ with $n$ vertices and edges, and an integer~$t$, in $2^{O(t)}n$ time, we either correctly report that the treewidth of~$G$ is at least~$t$, or compute a tree decomposition of width at most~$10t$.  In a first step, we apply it with $t=ck^d$.  If we get a tree decomposition of width at most $10ck^d$, we immediately obtain, and return, a (rooted) branch decomposition of width at most $15ck^d$, as desired.  The rest of the proof focuses on the former case, in which the treewidth of~$G$ is at least~$ck^d$.

  We will also use a result by Perkovic and Reed~\cite{pr-iaftd-00} (see also Sau, Stamoulis, and Thilikos~\cite[Proposition~28]{sst-kamcg-22}), which implies the following: Given an input graph~$H$ with $n$ vertices, we can compute in $O(n)$ time a minor~$H'$ of~$H$ with $|V(H')|\le\frac{15}{16}|V(H)|$ such that the treewidth of~$H'$ is at least half the treewidth of~$H$.  (This is because contracting a matching of a graph divides its treewidth by at most two.)
  
  We alternately apply the algorithm by Bodlaender et al.\ with $t=2ck^d$, and the algorithm by Perkovic and Reed, until Bodlaender's algorithm computes a tree decomposition of width at most $20ck^d$.  This must happen at some point because the algorithm of Perkovic and Reed decreases the number of vertices geometrically.  The step before this happens, the treewidth was at least~$2ck^d$, so at the end it is also at least~$ck^d$.  We thus computed, in $O(2^{O(k)}n)$ time, a minor~$G'$ of~$G$ whose treewidth is between $ck^d$ and $20ck^d$, together with a tree decomposition of~$G'$ of width at most~$20ck^d$.

  By the choice of $c$ and~$d$, the graph~$G'$ has a $(k\times k)$-grid minor, which we can compute in $O(2^{\poly(k)}n)$ time because we have a tree decomposition of width at most~$20ck^d$, using an algorithm by Adler, Dorn, Fomin, Sau, and Thilikos~\cite[Theorem~1]{adfst-fpamc-11}  (the theorem is stated as a decision problem, but in Section~5 of the same paper it is explained how to explicitly find a minor).
\end{proof}
Alternatively, with fewer and more standard tools, one could achieve this with a worse running time of $f(k)n^2$, for some computable function~$f$.  Indeed, we can first approximate the treewidth of~$G$, e.g., by the algorithm by Bodlaender et al.~\cite{bddflp-c5aat-16}: Provided $c$ and~$d$ are large enough, in time $2^{\poly(k)}n$, we either compute a tree decomposition of width at most $10ck^d$, and thus immediately obtain a (rooted) branch decomposition of width at most~$15ck^d$~\cite[Theorem~5.1]{rs-gm10o-91}, or correctly report that the treewidth of~$G$ is at least~$ck^d$, which by the result by Chekuri and Chuzhoy~\cite{cc-pbgmt-16} implies the existence of a $(k\times k)$-grid minor, which we can compute in $f(k)n^2$ time, for some computable function~$f$~\cite[Algorithm~4.4]{rs-gm10o-91}.

\def\csg{\lceil\sqrt g\rceil}
\begin{proof}[Proof of Proposition~\ref{P:irrelevant_0}]
  We apply Lemma~\ref{L:branch-or-grid} with $k=60{\lceil\sqrt g\rceil}^4$.  If the outcome is a rooted branch decomposition, then the algorithm returns it (Case~1).  Otherwise, we have computed a $(60{\lceil\sqrt g\rceil}^4\times60{\lceil\sqrt g\rceil}^4)$-grid minor of~$G$, and thus a subgraph~$\dot W$ of~$G$ that is a subdivision of a $(60{\lceil\sqrt g\rceil}^4\times60{\lceil\sqrt g\rceil}^4)$-wall~$W$.

  We first compute, in $O(g^4n)$ time, disjoint non-adjacent $(50g\csg\times 50g\csg)$-walls $W_1, \dots, W_g$ of~$W$, and the corresponding subdivisions $\dot W_1,\dots,\dot W_g$ that are subgraphs of~$\dot W$, in such a way that $W-W_i$ is connected for each~$i$.  For each~$i$, we consider the subgraph~$G_i$ of~$G$ induced by the vertices $v$ of~$G$ such that: (1) there exists a path from $v$ to~$\dot W_i$; (2) every path from $v$ to~$\dot W_j$, for some $j\ne i$, uses at least one vertex from~$\dot W_i$.  The graphs~$G_i$ are pairwise disjoint, and we can compute each of them easily in $O(n)$ time using suitable traversals of~$G$.  We test the planarity of each of them in linear time~\cite{ht-ept-74}.  If all these disjoint graphs are non-planar, we correctly report that $G$ has genus at least~$g$ (Case~2); this is easy to prove directly, and also follows from more general results, see, e.g., Miller~\cite[Theorem~1]{miller1987additivity}.  So without loss of generality, one of these graphs, say~$G_1$, is planar, and our algorithm computes it.

  By 3-connectivity, $W_1$ has a unique combinatorial embedding in the plane, up to symmetry and up to the choice of the outer (infinite) face; we consider the natural embedding of~$W_1$ in which the outer face has the largest degree.  An \emph{inner vertex} of~$W_1$ is one that is at distance, in~$W_1$, at least~6 from the vertices on the outer face in the natural embedding of~$W_1$ (we emphasize that this is the distance measured in the graph~$W_1$).  Remark that each vertex of~$W_1$ is a vertex of~$\dot W_1$.  We say that a vertex~$u$ of~$W_1$ is \emph{connected to the outside} if there is, in~$G$, a path whose vertices are, in this order, $u$, possibly some vertices of~$\dot W_1$ but not of~$W_1$, possibly some vertices of $G-\dot W_1$, and finally one vertex in~$\dot W-\dot W_1$.

  We claim the following: \emph{If at least $1000g$ inner vertices of~$W_1$ are connected to the outside, then $G$ has genus at least~$g$.}
  The strategy is similar to the argument in Kociumaka and Pilipczuk~\cite[Lemma~5.3]{kp-dvgbg-19}; we summarize the proof.  If at least $1000g$ inner vertices of~$W_1$ are connected to the outside, then a set~$U$ of $g$ inner vertices of~$W_1$ are connected to the outside, and at pairwise distance at least~16 in~$W_1$.  This implies that $G$ contains, as a minor, the graph~$J_g$ obtained from $g$ copies of~$K_5$ by subdividing an edge from each copy with a degree-two vertex and identifying these $g$~new vertices into a single vertex, the \emph{apex} of~$J_g$; see Figure~\ref{F:gadget}.  This graph has genus at least~$g$~\cite[Theorem~1]{miller1987additivity}.  This proves the claim.

\begin{figure}
  \centering
  \includegraphics[width=\linewidth]{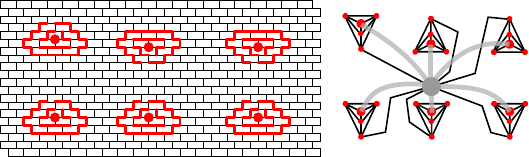}
  \caption{Illustration of the proof of the claim in Proposition~\ref{P:irrelevant_0}. The wall~$W_1$ (left), and a schematic view of the $J_p$ minor (right), illustrated with $g=6$.  Let $K_5^-$ be the graph $K_5$ with one edge removed.  If $W_1$ has many vertices connected to the outside, then a set~$U$ of at least $g$~vertices connected to the outside (represented as big disks on the left) are pairwise distant.  In the neighborhood of these vertices, we build a $K_5^-$ minor of~$\dot W_1$ for each vertex in~$U$, in which the vertex of~$U$ is the ``central'' vertex (the thick paths need to be contracted to obtain copies of~$K_5^-$).
  Because the vertices of~$U$ are pairwise distant, these minors are pairwise disjoint, and moreover the graph~$X$ consisting of the subdivided wall~$\dot W$ minus the union of these $K_5^-$ minors is connected.  From these $K_5^-$ minors, we build a $J_g$ minor of~$G$ (right) as follows.  The apex results from the contraction of~$X$.  Moreover, recall that each vertex of~$U$ is connected to the outside; each light thick edge, connecting some vertex $u\in U$ to~$X$, is obtained by contracting the corresponding path from~$u$ to a vertex of~$X$, except its first edge.}
\label{F:gadget}
\end{figure}
  In $O(g^3n)$ time, we can compute the inner vertices of~$W_1$ connected to the outside, using suitable traversals of~$G$ starting from $\dot W-\dot W_1$.  If there are at least $1000g$ of these, we report that $G$ has genus at least~$g$ (Case~2), which is correct by the above claim.  Otherwise, from the $(50g\csg\times 50g\csg)$-wall $W_1$, we can compute a cycle~$\gamma$ in~$W_1$, enclosing a $(g\times g)$-wall~$W'_1$ in~$W_1$ (in the natural embedding of~$W_1$) such that no vertex of~$W_1$ inside~$\gamma$ (in the natural embedding of~$W_1$) is connected to the outside.  Let $\dot\gamma$ be the cycle of~$\dot W_1$ corresponding to~$\gamma$, and let $H$ be the connected component of~$G-\dot\gamma$ containing the vertices of this $(g\times g)$-wall.  There remains to prove that $H$ is planar, and since $G_1$ is planar, it suffices to prove that $H$ is a subgraph of~$G_1$.  If it were not the case, $H$ would contain a vertex of~$\dot W_j$ for some $j\ne1$; but that would imply that a vertex of~$W'_1$ is connected to the outside, which is not the case.  We can thus correctly report~$\gamma$ (Case~3).
\end{proof}

\subsection{Finding an irrelevant vertex}

The following proposition will imply that if the third possibility in the statement of Proposition~\ref{P:irrelevant_0} holds (for some $g$ large enough), then one has an \emph{irrelevant vertex} for the embedding instance.
\begin{proposition}\label{P:irrelevant}
  Let $\complex$ be a 2-complex with $c\ge1$ simplices.  Let $G$ be a graph and $\gamma$ be a cycle in~$G$ such that one connected component of $G-\gamma$ is planar and contains a subdivision of the $(200c\times 200c)$-wall.  Let $v$ be the central vertex of this wall.  Then $G$ is embeddable on~$\complex$ if and only if $G-v$ is.
\end{proposition}

The proof is inspired from an article by Kociumaka and Pilipczuk~\cite[Section~5.3]{kp-dvgbg-19}, but we provide a slightly different argument, also because we need to handle a 2-complex, not a surface.  We use the topological concept of (free) homotopy, see, e.g., Stillwell~\cite{s-ctcgt-93}; for completeness, we recall the definition.  A \emphdef{closed curve} on a surface~$\surf$ is a continuous map from the unit circle~$S^1$ into~$\surf$.  Two closed curves $c_0$ and~$c_1$ are \emphdef{homotopic} if, intuitively, one can be deformed to the other or to its reversal; formally, this means that there is a continuous map $h:[0,1]\times S^1\to\surf$ such that for each $t\in[0,1]$, we have $h(0,t)=c_0(t)$, and either for each~$t$, $h(1,t)=c_1(t)$ or for each~$t$, $h(1,t)=c_1(1-t)$.
\begin{lemma}\label{L:freehomotopy}
  On a surface of genus~$g$, a set of disjoint, simple closed curves belong to at most $3g+2$ (free) homotopy classes (here we regard a closed curve and its reversal to be in the same homotopy class).
\end{lemma}
\begin{proof}
  Let $C$ be a set of disjoint, simple closed curves on a surface~$\surf$ of genus~$g$.  Assume that in $C$ there is no contractible curve, and there are no two distinct curves in the same homotopy class.  We prove below that $|C|\le 3g+1$, which proves the lemma.
  
  We first turn $C$, by homotopy, into a set of loops~$L$ pairwise disjoint except at a common basepoint~$b$.  To this end, we first draw a tree~$T$ on~$\surf$ that touches each closed curve in~$C$ exactly once.  (For this, one could start with a tree that is a spanning tree of the ``dual'' of the closed curves in~$C$ and visits all faces of~$C$, and then extend it if necessary to touch all curves in~$C$.)  We then contract~$T$ on~$\surf$. 

  $L$ cannot contain a monogon, because otherwise the corresponding closed curve in~$C$ would be contractible.  $L$ cannot contain a bigon whose boundary involves two different loops, because otherwise the corresponding closed curves in~$C$ would be homotopic.  If $L$ contains a bigon whose boundary involves the same loop appearing twice, then $|L|=1$, so the lemma holds.  By Lemma~\ref{L:nogon}, $|L|\le\max\{0,3g-3\}$.
\end{proof}

\begin{proof}[Proof of Proposition~\ref{P:irrelevant}]
  We need some preparations.  A \emph{circular wall}~$W$~\cite[Figure~8]{kp-dvgbg-19} of height~$h$ and circumference~$\ell$ is a 3-regular graph that, in some embedding of~$W$ in the plane, is represented as the union of $h$~vertex-disjoint cycles, called \emph{circles}, organized in a concentric way, such that any two consecutive circles are connected by $\ell$~\emph{radial edges}; in any two successive layers, the radial edges are interleaved.

  Let $H$ be the connected component of $G-\gamma$ that is planar and contains a subdivision of the $(200c\times 200c)$-wall.  By 3-connectivity, this wall has a unique combinatorial embedding in the plane, up to symmetry and up to the choice of the outer (infinite) face; we denote by $\varphi$ a plane embedding of~$H$ such that the restriction of~$\varphi$ to the wall is the natural embedding, in which the outer face of the wall has the largest degree.  In particular, $v$ is the central vertex of this wall.  
  
  Then $H$ contains a subdivision~$\dot W$ of a circular wall~$W$ of height $11c+4$ and circumference $2c+1$, so that $v$ is located inside the inner circle of this circular wall in $\varphi$; see Kociumaka and Pilipczuk~\cite[Figures 8 and~9]{kp-dvgbg-19}.

  If $G$ is embeddable on~$\complex$, then obviously $G-v$ is also embeddable on~$\complex$; the hard part is the reverse direction. So let us consider an embedding~$\psi$ of~$G-v$ on~$\complex$.  This induces embeddings of~$\dot W$ and~$W$ on~$\complex$.  At most $c$ vertices or edges of~$W$ are mapped, under~$\psi$, to a singular vertex of~$\complex$, and the isolated edges of~$\complex$ intersect the images of at most $c$~edges of~$W$.  Thus, we obtain a subdivision~$\dot W'$ of a circular wall~$W'$ of height $9c+4$ and circumference one that does not intersect any singular vertex or isolated edge of~$\complex$.

  Let $C_1,\ldots,C_{9c+4}$ be the cycles in~$G-v$ that correspond to the circles of~$W'$, such that in~$\varphi$, they are nested in this order, $C_1$ being the innermost cycle.  Let $C$ be the union of these cycles.  A \emph{bridge} (of~$C$) is either (1) a connected component of $G-C$, together with the edges joining it to~$C$, or (2) an edge of~$G-C$ whose endpoints are on~$C$.  For a bridge of~$C$, the vertices of~$C$ that are endpoints of edges in the bridge are the \emph{attachment points} of the bridge.  For each integer~$i$, let $B_i$ be the set of bridges all of whose attachment points belong to~$C_i$; in particular, $v$ belongs to a bridge in~$B_1$.  Moreover, for each integer~$i$, let $B_{i+\half}$ be the set of bridges having at least one attachment point on~$C_i$ and at least one on~$C_{i+1}$.  The existence of~$\varphi$ implies that there are no other bridges.

  We claim that, for some integer~$i_1$, the cycle~$C_{i_1}$ is two-sided and the edges of the bridges in~$B_{i_1+\half}$ leave~$C_{i_1}$, in~$\psi$, all on the same side of~$C_{i_1}$.  We prove the claim in the next few paragraphs.

  Let $\surf$ be the detached surface of~$\complex$, and let $\hat\surf$ be the surface obtained by attaching a handle to every boundary component of~$\surf$ and to a neighborhood of each point of~$\surf$ not on its boundary that corresponds, in~$\complex$, to a singular vertex.  At most $2c$ handles were attached, so $\hat\surf$ has a natural cellular graph embedding with at most $3c$~edges, and thus has genus at most~$3c$.  Moreover, $\psi$ naturally corresponds to an embedding of $\dot W'$ and~$W'$ on~$\hat\surf$, which we still denote by~$\psi$.

  By Lemma~\ref{L:freehomotopy}, since $\hat\surf$ has genus at most~$3c$, and $W'$ has height $(3*3c+2)+2$, there are two cycles $C_{i_0}$ and $C_{i_1}$, with $i_0<i_1<i_2:=(3*3c+2)+2$, that are homotopic in~$\hat\surf$ under~$\psi$.  They are two-sided, because any two one-sided homotopic closed curves cross.  To prove our claim, we distinguish two cases, implicitly using the fact that, as testified by~$\varphi$, each bridge in $B_{i_1+\half}$ is connected to~$C_{i_2}$ in $G-(C_{i_0}\cup C_{i_1})$:
  \begin{itemize}
    \item Assume that $\psi(C_{i_1})$ is contractible on~$\hat\surf$; it thus bounds a disk on~$\hat\surf$, and thus separates~$\complex$ into two pieces.  Thus, the edges of the bridges in~$B_{i_1+\half}$ leave~$C_{i_1}$, in~$\psi$, on the side of~$C_{i_1}$ corresponding to the piece that contains~$C_{i_2}$, as desired.
    
    \item Otherwise, the closed curves $\psi(C_{i_0})$ and $\psi(C_{i_1})$ are homotopic on~$\hat\surf$ but non-contractible, and thus bound an annulus in~$\hat\surf$~\cite[Lemma~2.4]{e-c2mi-66}, and thus also in~$\complex$; we denote by~$A$ the annulus on~$\complex$.  There are two subcases:
    \begin{itemize}
      \item Assume that, under~$\psi$, the annulus~$A$ does not contain~$C_{i_2}$.  Then, $B_{i_1+\half}$ cannot enter~$A$, which implies the claim;
      \item otherwise, under~$\psi$, $B_{i_1+\half}$ must lie entirely inside~$A$, which also implies the claim.
    \end{itemize}
  \end{itemize}
  So the claim is proved.
  
  To conclude the proof, we first consider the restriction~$\psi'$ of~$\psi$ to $G\setminus (C_{\le i_1-1}\cup B_{\le i_1})$.  (For each integer~$i$, $C_{\le i}$ denotes the union of the cycles $C_j$, $1\le j\le i$; similarly, for each integer or half-integer~$i$, $B_{\le i}$ denotes the union of the bridges $B_j$, $j$ being an integer or a half-integer between 1 and~$i$.)
  
  Next, we extend $\psi'$ to an embedding~$\psi''$ of~$G\setminus B_{i_1}$ into~$\complex$, by mapping $C_{\le i_1-1}\cup B_{\le i_1-\half}$ close to $C_{i_1}$, on the side of~$C_{i_1}$ that is not already used, as in~$\varphi$. 

  Finally, there remains to extend~$\psi''$ to~$B_i$.  For this, we note that, under~$\varphi$, the set of edges of~$G$ incident to~$C_{i_1}$ is divided into two parts according to the side of~$C_{i_1}$ they leave~$C_{i_1}$, and that, by construction, this partition is the same under~$\psi''$.  We can thus augment~$\psi''$ to an embedding of all of~$G$ into~$\complex$, by embedding~$B_{i_1}$ close to~$C_{i_1}$ as it is done in~$\varphi$.
\end{proof}

\subsection{Proof of Theorem~\ref{T:irrelevant}}

\begin{proof}[Proof of Theorem~\ref{T:irrelevant}]
  We first apply Proposition~\ref{P:preproc}: without loss of generality, $\complex$ has no 3-book and no connected component that is reduced to a single vertex, and $G$ has no connected component reduced to a single vertex, and at most one connected component that is a path.  Let $n$ be the number of vertices and edges of the input graph~$G$, and $c$ be the number of simplices of~$\complex$.  We apply Proposition~\ref{P:irrelevant_0} to the graph~$G$, letting $g=200c$.  In $2^{\poly(c)}\cdot n$ time, we obtain one of the following outcomes:
  \begin{enumerate}
      \item a rooted branch decomposition of~$G$ of width $g^{O(1)}$;
      \item that $G$ has genus at least $g$, and is thus not embeddable on~$\complex$ (Lemma~\ref{L:genus-oversurface});
      \item a cycle~$\gamma$ of~$G$ such that one connected component of $G-\gamma$ is planar and contains a subdivision of the $(g\times g)$-wall.
  \end{enumerate}
  In the first two cases, we are done.  In the third case, by applying Proposition~\ref{P:irrelevant}, we obtain a vertex~$v$ such that $G$ embeds on~$\complex$ if and only if $G-v$ does.  By iterating the same procedure a number of times that is at most the number of vertices of~$G$, we necessarily reach case (1) or~(2), which concludes.
\end{proof}
We remark that the proof goes through if the input 2-complex is given in the form of the topological data structure, and $c$ denotes its size, instead of the number of simplices of~$\complex$.

As mentioned above, the proof of Theorem~\ref{T:general} follows immediately from Theorems~\ref{T:bounded-bw} and~\ref{T:irrelevant}.

\section{Computing an embedding}\label{S:computing}

In this section, we enhance Theorems \ref{T:general} and~\ref{T:bounded-bw} to show that, under some mild conditions and in a sense made precise below, if the input graph~$G$ has an embedding into the input complex~$\complex$, then we can actually compute such an embedding without overhead in the asymptotic running time:

\begin{theorem}\label{T:compute}
  In Theorems \ref{T:general} and~\ref{T:bounded-bw}, if $\complex$ has no 3-book and no isolated vertex, and $G$ embeds into~$\complex$, then one can, without overhead in the asymptotic running time, compute an embedding of a graph~$H$ on~$\complex$ where:
  \begin{itemize}
      \item $H$ is obtained from~$G$ by augmenting it with at most~$2c$ vertices and at most $3c+2k$ edges, and performing at most $c$ edge subdivisions, where $k$ is the number of connected components of~$G$;
      \item the images of the vertices of~$H$ cover the singular vertices of~$\complex$;
      \item the restriction of~$H$ to the detached surface~$\surf$ is specified by (1) a cellular embedding of a graph~$H'$, represented by its (standard) combinatorial map, such that $E(H')\subseteq E(H)$, (2) a map from $V(H')$ to~$V(H)$ describing how the vertices of~$H'$ are identified when attaching the points of~$\surf$ to recover the singular vertices of~$\complex$, and (3) a map from the images of the singular vertices of~$\complex$ on~$\surf$ to $V(H')$, to specify which vertex of~$H'$ occupies that singular vertex.  Moreover, we have $|V(H')|\le |V(H)|+c$;
      \item the image of each isolated edge~$e$ of~$\complex$ is the disjoint union of the images of some vertices and edges of~$H$, and of some connected parts of~$e$ not in the image of~$H$; we can represent these objects by listing them in order along~$e$.
  \end{itemize}
\end{theorem}

We remark that it is not a severe restriction to assume that $\complex$ has no 3-book, since any graph embeds on a 3-book (Proposition~\ref{P:preproc}), and no isolated vertex, since only isolated vertices of~$G$ can be embedded on such isolated vertices.  Also, we remark that the graph~$H$ is needed because we may need to add edges to make the embedding cellular on the detached surface, and we may need to subdivide edges in order to avoid any edge to go through a singular vertex.

We say that a graph~$G'$ is an \emphdef{$i$-homeomorph} of another graph~$G$ if $G'$ is obtained from~$G$ by the following steps: (i) iteratively dissolve every degree-two vertex of~$G$, (2) subdivide each edge of the resulting graph at most $i$~times.

We will use the following intermediate lemmas.

\begin{lemma}\label{L:compute-proper-bounded-bw}
  In Theorem~\ref{T:bounded-bw}, if $\complex$ has no 3-book and no isolated vertex, $G$ has no isolatex vertex and at most one connected component that is a path, and $G$ embeds into~$\complex$, then one can compute a proper embedding of a $5c$-homeomorph of~$G$, augmented with at most $c$ isolated vertices (using the data structure of Section~\ref{S:datastruct-graphs}), without overhead in the asymptotic running time.
\end{lemma}
\begin{proof}
  We go through the proof of Theorem~\ref{T:bounded-bw} and express what can be computed for a positive instance $(G,\complex)$ of the embeddability problem.  Let us first recapitulate the algorithm, omitting the step of Proposition~\ref{P:preproc} because of the assumptions on~$\complex$ and~$G$.

  \begin{enumerate}  
  \item[1.] In the first step (Proposition~\ref{P:cell}), from $(G,\complex)$, we compute a graph~$G'$ and complexes $\complex_i$, such that any embedding of~$G'$ into~$\complex_i$ yields an embedding of~$G$ into~$\complex$.

  \item[2.] In the second step, we run the dynamic program (Proposition~\ref{P:dynprog}) on $(G',\complex_i)$, for each complex~$\complex_i$.
  \end{enumerate}

  We now consider these two steps in reverse order.
  \begin{enumerate}
  \item[2.] Since $G$ embeds on~$\complex$, for some complex~$\complex_i$, we have computed a sparse bounding graph for~$G'$ on~$\complex_i$ at the root of the rooted branch decomposition.  We can use backtracking: at each node~$\alpha$ of the rooted branch decomposition, we store the embedding of the union of~$G_\alpha$ and of its bounding graph, which is a proper embedding.  Eventually, the embedding of~$G'$ at the root of the recursion is not necessarily proper, but it becomes proper if we add isolated vertices mapped to singular vertices that are not covered by the embedding of~$G'$.  Thus, from $(G',\complex_i)$, we can return a proper embedding of the graph~$G'$, augmented with at most $c$ isolated vertices, into~$\complex_i$.

  \item[1.] An easy inspection of its proof reveals that:
  \begin{itemize}
  \item $G'$ is obtained from~$G$ by removing planar connected components, but only if $\complex$ has at least one triangle, and then by dissolving iteratively all degree-two vertices and subdividing every edge $5c$ times;
  \item $\complex_i$ is obtained from~$\complex$ by removing part of its isolated edges, splitting singular vertices, performing essential cutting operations, and removing some connected components.
  \end{itemize}
  Given a proper embedding of~$G'$, augmented with at most $c$ isolated vertices, into~$\complex_i$, we can embed the remaining planar components of~$G$ in (a small portion of a triangle) of~$\complex_i$, revert the pieces of the isolated edges of~$\complex$, merge back the singular vertices (if necessary, adding a vertex to the graph on that singular vertex to maintain the fact that we have a proper embedding), and reverse the essential cutting operations on the complex.  We obtain a proper embedding of a $5c$-homeomorph of~$G$, augmented with at most $c$ isolated vertices, into~$\complex$.
  \end{enumerate}

  Moreover, at every step, there is no overhead in the asymptotic running time.
\end{proof}

We have the same constructive result for our Theorem~\ref{T:general}:
\begin{lemma}\label{L:compute-proper-general}
  In Theorem~\ref{T:general}, if $\complex$ has no 3-book and no isolated vertex, $G$ has no isolated vertex and at most one connected component that is a path, and $G$ embeds into~$\complex$, then one can compute a proper embedding of a $5c$-homeomorph of~$G$, augmented with at most $c$ isolated vertices, into~$\complex$ (using the data structure of Section~\ref{S:datastruct-graphs}), without overhead in the asymptotic running time.
\end{lemma}
\begin{proof}
  The decision algorithm iteratively removes vertices to~$G$, getting a smaller equivalent instance.  For the algorithm that actually computes an embedding, we now need recursion.  The base case, for bounded branchwidth, is given by Lemma~\ref{L:compute-proper-bounded-bw}.

  In the general case, we have computed a cycle~$\gamma$ of~$G$ such that one connected component of~$G-\gamma$ is planar and contains a subdivision of the $(200c\times200c)$-wall; this wall, and the center~$v$ of this wall, have also been computed (see Proposition~\ref{P:irrelevant}).  We assume that we have a proper embedding~$\psi$ of a $5c$-homeomorph of~$G-v$, augmented with at most $c$ isolated vertices, into~$\complex$.  We need to recover a proper embedding of a $5c$-homeomorph of~$G$, augmented with at most $c$ isolated vertices, into~$\complex$.

  The algorithm follows the proof and notation of Proposition~\ref{P:irrelevant}.  Starting from the above wall, we compute the subdivision~$\dot W$ of the corresponding circular wall~$W$.  From the embedding~$\psi$, we obtain a subdivision~$\dot W'$ of a circular wall~$W'$ of height $9c+4$ and circumference one that does not intersect any singular vertex or isolated edge of~$\complex$.  In linear time, we compute~$\varphi$, and also the bridges of~$C$ and their partition into the sets~$B_i$.

  In the proof of Proposition~\ref{P:irrelevant}, we proved that, for some~$i_1$, the cycle~$C_{i_1}$ is two-sided and the edges of the bridges in~$B_{i_1+\half}$ leave $C_{i_1}$, in~$\psi$, all on the same side of~$C_{i_1}$.  From the $B_i$s, we can easily compute such a value of~$i_1$ in linear time.  Once this is done, as in the proof of Proposition~\ref{P:irrelevant}, from the embeddings~$\varphi$ and~$\psi$ we can compute a proper embedding of a $5c$-homeomorph of~$G$, augmented with at most $c$ isolated vertices, into~$\complex$; this takes linear time.
\end{proof}

\begin{proof}[Proof of Theorem~\ref{T:compute}]
  Let $G'$ be obtained from~$G$ by merging its isolated vertices and path components into a single path component, as in the proof of Proposition~\ref{P:preproc}.  We apply Lemma \ref{L:compute-proper-bounded-bw} or~\ref{L:compute-proper-general}, obtaining a proper embedding of a $5c$-homeomorph of~$G'$, augmented with at most $c$ isolated vertices, into~$\complex$ (using the data structure of Section~\ref{S:datastruct-graphs}).  

  We now apply the following transformations.  First, for each connected component of the detached surface~$\surf$ that is not used by the embedding, we add a vertex of the graph and map it to that connected component.  Then, on each connected component of~$\surf$, we make the embedding connected by adding at most $2c+k$ edges in total; this is because, on~$\surf$, the image of the number of connected components of~$G'$ is at most $2c+k$.  In a third step, we make the embedding on~$\surf$ cellular by adding edges; by Euler's formula, the number of edges that are needed is at most the Euler genus of~$\surf$, which is at most~$c$.

  We now have a proper embedding of a $5c$-homeomorph of~$G'$, augmented with at most~$2c$ vertices and at most $3c+k$ edges, into~$\complex$, satisfying the property that this embedding, restricted to the detached surface~$\surf$, is cellular.  If necessary, we subdivide further the edges to obtain a proper embedding of the graph~$G'$, in which each edge is subdivided at most~$5c$ times, and then augmented with at most~$2c$ vertices and at most $3c+k$ edges.  We dissolve degree-two vertices that are not part of~$G$ and not on singular vertices.  There remains, if necessary, to subdivide the edge of~$G'$ corresponding to the path component of~$G$ at most $2k$ times, so that it now contains all the path components and isolated vertices of~$G$, in order to obtain the desired embedding; this adds at most $k$ edges that are not part of~$G$.

  Eventually, we have a graph~$H$ obtained from~$G$ by augmenting it with at most~$2c$ vertices and at most $3c+2k$ edges, and performing at most $c$ edge subdivisions (the corresponding vertices being on singular vertices of~$\complex$).  Moreover, the images of the vertices of~$H$ cover the singular vertices of~$\complex$.  Also, $H$ corresponds to a cellular embedding on the detached surface~$\surf$, so it can be encoded by a combinatorial map as claimed, and similarly one can encode the embedding~$H$ restricted to the isolated edges of~$\complex$.
\end{proof}

\section{Solving other problems}\label{S:applications}

In this short section, we prove (or reprove) results that are easy consequences of our main result.  We focus on decision problems; in the case of positive instances, explicit constructions can be done using Theorem~\ref{T:compute} without overhead in the runtime.  As we describe in Section~\ref{S:conclusion}, the decision problems can also be solved in time linear in~$n$, but with an unspecified dependency in the parameter, by combining our techniques with Mohar's linear-time algorithm to embed graphs into surfaces~\cite{m-ltaeg-99} together with a recent algorithm for the irrelevant vertex method~\cite{gkst-fivlt-25}.

\subsection{Crossing number problem}\label{S:crossing-number}

Recall that the \emphdef{crossing number problem} is to decide, given a graph~$G$ with at most $n$ vertices and edges and an integer~$k$, whether there is a (topological) drawing of~$G$ in general position in the plane with at most $k$ crossings.  (General position means that the only crossings involve exactly the relative interiors of two edges.)

The problem is known to be NP-hard, and Kawarabayashi and Reed, in an extended abstract~\cite{kr-ccnlt-07}, provided a linear-time algorithm if $k$ is fixed.  Our result immediately implies a weaker, quadratic-time algorithm:
\begin{proposition}
  We can solve the crossing number problem in $O(2^{\poly(k)}n^2)$ time.
\end{proposition}
\begin{proof}
  The reduction is known~\cite[Introduction]{cmm-egtds-22} but we include it here for completeness.  Consider the topological space that is a disk minus $k$~open disks with disjoint closures.  To the boundary of each of these $k$ smaller disks, attach two edges whose endpoints are interleaved along the boundary.  This topological space is homeomorphic to a two-dimensional simplicial complex~$\complex$ with $O(k)$ simplices.

  It is clear that any graph~$G$ embedded in~$\complex$ has crossing number exactly~$k$.  Conversely, consider a graph drawn in the plane in general position with at most $k$ crossings; by removing small disks around the crossings, we see that the graph is embedded on a topological space homeomorphic to~$\complex$.  Theorem~\ref{T:general} concludes the proof.
\end{proof}
Actually, using Theorem~\ref{T:bounded-bw}, we obtain an algorithm that in $(k+w)^{O(k+w)}n$ time solves the crossing number problem when restricted to graphs of branchwidth at most~$w$, which is a new result; see Lokshtanov, Panolan, Saurabh, Sharma, Xue, and Zehavi~\cite{lpssxz-cnsst-25}.

Let us point out that we can handle more general situations.  For example, using the same technique, we immediately obtain an algorithm to decide whether an input graph has a drawing on a given surface of genus~$g$ with at most $k$ crossings, in $2^{\poly(g+k)}n^2$ time.

\subsection{Embedding extension problem}

An \emphdef{embedding extension problem} takes as input a 2-complex~$\complex$ of size~$c$, a graph~$G$ with at most $n$ vertices and edges, a subset~$U\subseteq V(G)$ of $k$~vertices of~$G$, and an embedding of~$U$ into~$\complex$.  The question is to determine whether there is an embedding of~$G$ into~$\complex$ extending the given embedding of~$U$.  We provide below an algorithm in a special case.

We remark that, in general embedding extension problems, not only a set of vertices is pre-embedded, but also possibly a set of edges.  In particular, Mohar~\cite{m-ltaeg-99} used embedding extension problems on surfaces as a subroutine to his linear time algorithm to embed graphs on a fixed surface; in the planar case, Angelini, Di Battista, Frati, Jel{\'\i}nek, Kratochv{\'\i}l, Patrignani, and Rutter~\cite{angelini2015testing} solve such embedding extension problems in linear time.  However, in our case, we can easily force the location of the edges as follows: Given a set of pre-embedded edges, this amounts to cutting the 2-complex along their images and to solving the embedding extension problem for pre-embedded vertices (the set of pre-embedded vertices includes the endpoints of the pre-embedded edges).

We provide an algorithm for the embedding extension problem on complexes without 3-book:
\begin{proposition}\label{P:EEP}
  We can solve the embedding extension problem for complexes without 3-book in $O(2^{\poly(c+k)}n^2)$ time.
\end{proposition}
\begin{proof}
  Without loss of generality, we can assume that the input complex~$\complex$ has no isolated vertex (see the proof of Proposition~\ref{P:preproc}).  Let $u_1,\ldots,u_k$ be vertices of the input graph~$G$, and $p_1,\ldots,p_k$ be distinct points of the input complex~$\complex$, such that $u_i$ must be embedded at location~$p_i$.

  We augment~$\complex$ into a larger complex~$\complex'$, and $G$ into a larger graph~$G'$, as follows.   For each $i$, let $H_i$ be a 2-connected graph of minimum degree at least three with $O(c+i)$ vertices and edges that is embeddable on the orientable surface of genus $10c+2i$ but not on any surface (orientable or not) of smaller Euler genus.  (It suffices to build graphs embedded on the surface with the prescribed genus with triangular faces, such that any non-contractible cycle has length at least four, by a result of Thomassen~\cite{t-egsnc-90}, and this is easy; see Figure~\ref{F:lew}.)  For each~$i$, let $h_i$ be an arbitrary but fixed vertex of~$H_i$.  We take a copy~$\complex_i$ of~$H_i$ and attach its copy of~$h_i$ to point~$p_i$ of~$\complex$.  Let $\complex'$ be the resulting space, which can be represented as a complex made of $O(k(c+k))$ simplices.  Similarly, for each~$i$, we take a copy~$G_i$ of~$H_i$ and attach it to~$G$, via the copy of~$h_i$, to vertex~$u_i$ of~$G$.  Let $G'$ be the resulting graph.
  
\begin{figure}
\centering
\includegraphics[width=.5\linewidth]{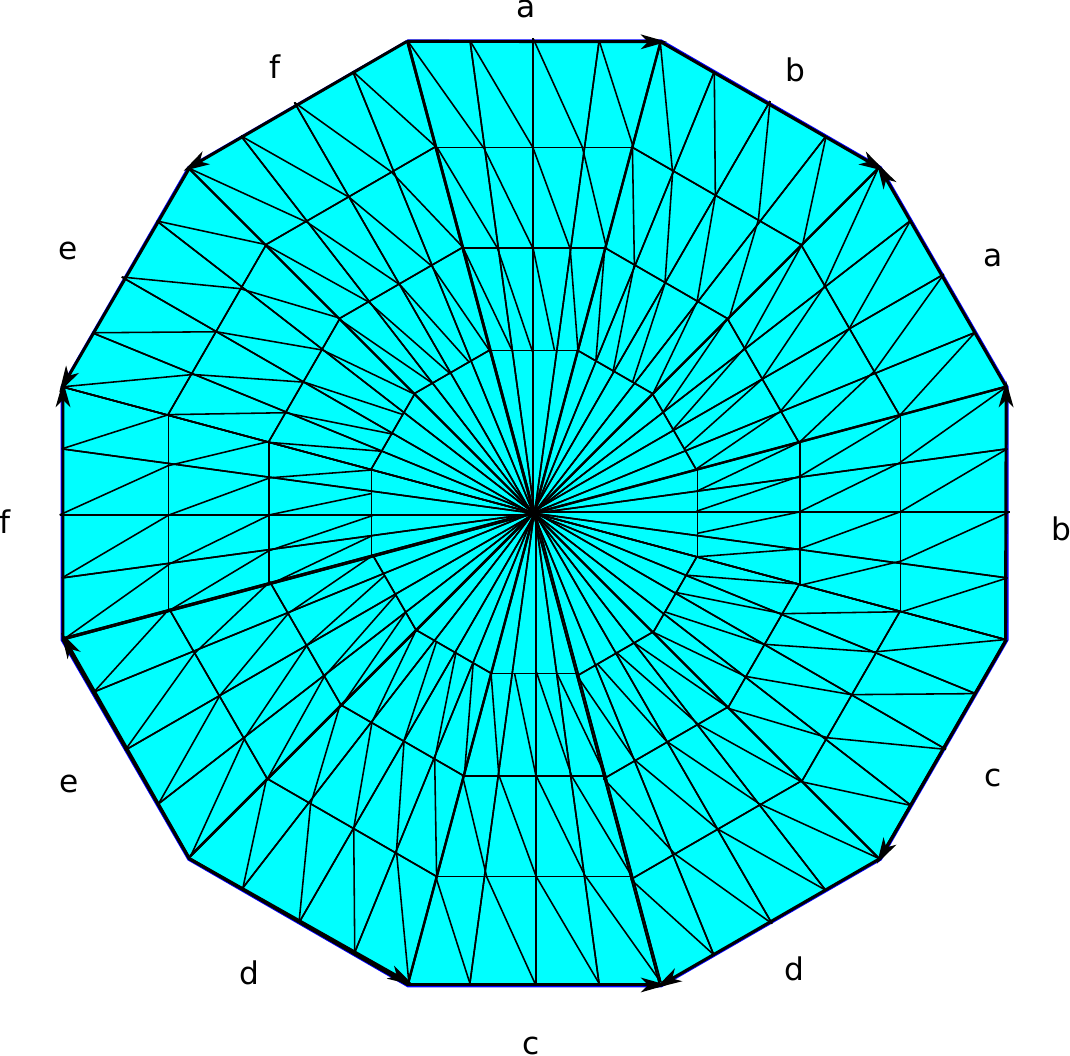}
\caption{A graph embedded on an orientable surface of Euler genus six, but not embeddable on a surface of lower genus.  Indeed, it is a triangulated graph in which all cycles of length at most three are contractible.}
\label{F:lew}
\end{figure}

  We claim that the instance of the embedding extension problem is positive if and only if $G'$ is embeddable in~$\complex'$, which concludes.  Clearly, if the instance of the embedding extension problem is positive, then $G'$ is embeddable in~$\complex'$:  Indeed, we embed $G$ into~$\complex$ as given by the solution of the embedding extension problem, and each graph~$G_i$ to~$\complex_i$.

  Conversely, assume that $G'$ is embeddable in~$\complex$.  By induction on $i=k,k-1,\ldots, 1$, since $H_i$ is 2-connected, $G_i$ must be embedded either in $\complex_i$ or in~$\complex$, and it must actually be in~$\complex_i$ by Lemma~\ref{L:genus-oversurface} and because $\complex$ has no 3-book, and finally must occupy exactly~$\complex_i$ because it has minimum degree at least three.  Necessarily, $u_i$ is mapped to~$p_i$, except perhaps in the very special case where $u_i$ has degree zero in~$G$, but in that case one can clearly modify the embedding of~$G_i$ on~$\complex_i$ to ensure that property.  Finally, taking the restriction of the embedding to~$G$ gives the desired embedding.
\end{proof}  

Handling 3-books looks substantially more complicated, and we leave it as an open problem.

\subsection{Planarity number problem}

Recall that the \emphdef{planarity number problem} is, given a graph~$G$ with at most $n$ vertices and edges, a subset~$U$ of vertices of~$G$, and an integer~$k$, to decide whether there exists a planar embedding of~$G$ in which a set of at most $k$~faces cover all the vertices in~$U$.  

Bienstock and Monma~\cite{bm-ccvfp-88} proved that this problem is strongly NP-hard, and provided a linear-time algorithm for fixed~$k$.  Our result immediately implies a weaker, quadratic-time algorithm:
\begin{proposition}
  We can solve the planarity number problem in $O(2^{\poly(k)}n^2)$ time.
\end{proposition}
\begin{proof}
  Let $T$ be the topological space obtained from the sphere by identifying $k$ points into into a single point, the \emph{apex}.  Let $G'$ be the graph~$G$ augmented with a single vertex~$b$, connected to all the vertices in~$U$.  It is straightforward to check that $(G,U)$ has planarity number at most~$k$ if and only if $G'$ has an embedding in~$T$ in which $b$ is mapped to the apex.  Proposition~\ref{P:EEP} concludes.
\end{proof}
As above, we can trivially generalize this result by replacing, in the definition of the planarity number problem, the plane with any fixed surface.

\section{Conclusion}\label{S:conclusion}

We conclude with a discussion on possible improvements on the running time of our algorithm.

First, as mentioned in the introduction, by combining our techniques with an algorithm to embed graphs on surfaces by Mohar~\cite{m-ltaeg-99} and a very recent algorithm for the irrelevant vertex method by Golovach, Kolliopoulos, Stamoulis, and Thilikos~\cite{gkst-fivlt-25}, we are able to get an algorithm whose running time has a linear dependency in~$n$ for our problem, at the cost of an increase in the parameter dependency:
\begin{theorem}\label{T:general-linear}
  One can solve the embeddability problem of graphs into 2-dimensional simplicial complexes in time $f(c)\cdot O(n)$ for some (unspecified) function~$f$, where $c$ is the number of simplices of the input 2-complex and $n$ is the total number of vertices and edges of the input graph.
\end{theorem}
In turn, this implies algorithms that are linear in~$n$ (with an unspecified dependence in the parameter) in all applications of Section~\ref{S:applications}.  It is, however, less clear whether one can obtain a similar runtime for computing actual embeddings (Theorem~\ref{T:compute}).

\begin{proof}
  We show that Theorem~\ref{T:irrelevant} holds with a runtime of $f(c)\cdot O(n)$.  The result immediately follows by an application of Theorem~\ref{T:bounded-bw}.  So we consider a 2-complex~$\complex$ with $c$ simplices, and a graph~$G$ with $n$ vertices and edges in total.  As we show next, the main result by Golovach et al.~\cite{gkst-fivlt-25} applies to our problem and implies that one can find, in time $f(c)\cdot O(n)$, a set~$I$ of vertices of~$G$ such that $G-I$ has branchwidth $O(c^{5/2})$, and such that $G$ embeds on~$\complex$ if and only if $G-I$ does.  This implies the theorem.

  The result by Golovach et al.\ requires an embedding of~$G$ on a surface of bounded genus.  For this purpose, using Mohar's algorithm~\cite{m-ltaeg-99}, we compute an embedding of~$G$ on a surface~$\bar\surf$ of genus $g\le10c$.  If such an embedding does not exist, then by Lemma~\ref{L:genus-oversurface}, we know that $G$ does not embed on~$\complex$.

  The result by Golovach et al.\ applies to problems that satisfy the so-called \emph{insulation property}.  We show that it indeed holds in our problem (with no roots).  So we assume that $G$ contains a so-called \emph{$r$-railed nest} for $r=11c+4$, namely, cycles $C_1,\ldots, C_r$ and paths $P_1,\ldots, P_r$, such that, in the embedding of~$G$ on some surface~$\bar\surf'$, the cycles $C_i$ bound disks that are nested in this order, $C_1$ being the innermost cycle, and such that for every $i,j$, the cycle~$C_i$ and the path~$P_j$ intersect in a (possibly trivial) path.  We have to prove that if $G-J$ embeds on~$\complex$, then so does $G$, where $J$ is the set of vertices inside~$C_1$ (in the embedding of~$G$ on~$\bar\surf'$).

  The proof is an easy variation on that of Proposition~\ref{P:irrelevant}.  By removing at most $2c$ cycles and paths, we obtain that there is a circular wall of height $9c+4$ and circumference one that, in~$\complex$, does not intersect any singular vertex or isolated edge.  Moreover, if $C_{9c+4}$ is the outer cycle, then the connected component of $G-C_{9c+4}$ that contains $J$ and the rest of the circular wall is planar.  The only difference in the rest of the proof is that $J$ is not necessarily a single vertex~$v$, but this does not affect the validity of the arguments.
\end{proof} 

Finally, a second direction would be to improve the $2^{\poly(c)}$ dependency while keeping the quadratic dependency in~$n$.  Our proof yields a $2^{O(c^{18+\varepsilon})}$ dependency; $18+\varepsilon$ is twice the exponent in the best known bound for the excluded grid theorem~\cite{ct-ttbeg-21}.  In particular, any improvement in that bound would improve the degree of the polynomial in~$c$ in our algorithm.  However, obtaining a much lower degree seems out of reach with our current techniques.

\subsection*{Acknowledgments}

We would like to thank Petr Hlin\v{e}n\'y, Arnaud de Mesmay, Giannos Stamoulis, and Dimitrios Thilikos for useful discussions, and the anonymous reviewers for their detailed comments.

\end{document}